\documentclass{article}
\PassOptionsToPackage{numbers, compress}{natbib}

\usepackage{arxiv}

\usepackage[utf8]{inputenc} 
\usepackage[T1]{fontenc}    
\usepackage{hyperref}       
\usepackage{url}            
\usepackage{amsthm,amsmath,amssymb,amsfonts,bbm,mathrsfs} 
\usepackage{nicefrac}       
\usepackage{microtype}      
\usepackage{subfigure}

\usepackage{graphicx}
\usepackage{natbib}
\usepackage{doi}
\usepackage{enumitem,kantlipsum}
\usepackage[capitalize,noabbrev]{cleveref}
\usepackage{algorithmic,algorithm}
\usepackage{booktabs, multirow}
\usepackage{rotating}
\usepackage{diagbox}
\usepackage{verbatimbox}
\usepackage{makecell}

\theoremstyle{plain}
\newtheorem{theorem}{Theorem}[section]
\newtheorem{proposition}[theorem]{Proposition}
\newtheorem{lemma}[theorem]{Lemma}
\newtheorem{corollary}[theorem]{Corollary}
\theoremstyle{definition}

\theoremstyle{remark}

\DeclareMathOperator*{\argmax}{arg\,max}
\DeclareMathOperator*{\argmin}{arg\,min}
\DeclareMathOperator*{\proj}{proj}
\newcommand\upquote[1]{\textquotesingle#1\textquotesingle}
\newcommand{\pro}{\mathbb P}

\newcommand{\ie}{\textit{i.e.}}
\newcommand{\eg}{\textit{e.g.}}

\title{A Learning and Control Perspective for Microfinance}

\date{} 					

\author{Christian Kurniawan$^+$\\
    Department of Electrical and Computer Engineering\\
    Carnegie Mellon University\\
    Pittsburgh, Pennsylvania, United States\\
    \texttt{christian.paryoto@gmail.com}\\
    \And
    Xiyu Deng$^+$\\
    Department of Electrical and Computer Engineering\\
    Carnegie Mellon University\\
    Pittsburgh, Pennsylvania, United States\\
    \texttt{xiyud@andrew.cmu.edu}\\
    \AND
    Adhiraj Chakraborty\\
    Carnegie Mellon University Silicon Valley\\
    California, United States\\
    \texttt{adhirajc@alumni.cmu.edu}\\
    \And
    Assane Gueye\\
    Carnegie Mellon University Africa\\
    Kigali, Rwanda\\
    \texttt{assaneg@andrew.cmu.edu} \\
    \And
    Niangjun Chen\\
    Information Systems Technology and Design\\
    Singapore University of Technology and Design\\
    Singapore\\
    \texttt{niangjun\_chen@sutd.edu.sg} \\
    \And
    Yorie Nakahira\\
    Department of Electrical and Computer Engineering\\
    Carnegie Mellon University\\
    Pittsburgh, Pennsylvania, United States\\
    \texttt{ynakahir@andrew.cmu.edu} \\
}

\renewcommand{\shorttitle}{A Learning and Control Perspective for Microfinance}

\hypersetup{
pdftitle={A Learning and Control Perspective for Microfinance},
pdfsubject={},
pdfauthor={},
pdfkeywords={Microfinance, Control System, Learning},
}

\begin{document}
\maketitle
\def\thefootnote{+}\footnotetext{These authors contributed equally to this work.}\def\thefootnote{\arabic{footnote}}

\begin{abstract}
	Microfinance, despite its significant potential for poverty reduction, is facing sustainability hardships due to high default rates. Although many methods in regular finance can estimate credit scores and default probabilities, these methods are not directly applicable to microfinance due to the following unique characteristics: a) under-explored (developing) areas such as rural Africa do not have sufficient prior loan data for microfinance institutions (MFIs) to establish a credit scoring system; b) microfinance applicants may have difficulty providing sufficient information for MFIs to accurately predict default probabilities; and c) many MFIs use group liability (instead of collateral) to secure repayment. Here, we present a novel control-theoretic model of microfinance that accounts for these characteristics. We construct an algorithm to learn microfinance decision policies that achieve financial inclusion, fairness, social welfare, and sustainability. We characterize the convergence conditions to Pareto-optimum and the convergence speeds. We demonstrate, in numerous real and synthetic datasets, that the proposed method accounts for the complexities induced by group liability to produce robust decisions before sufficient loans are given to establish credit scoring systems and for applicants whose default probability cannot be accurately estimated due to missing information. To the best of our knowledge, this paper is the first to connect microfinance and control theory. We envision that the connection will enable safe learning and control techniques to help modernize microfinance and alleviate poverty.

\end{abstract}


\section{Introduction}
\label{S:Introduction}
\textbf{Potential and challenges in microfinance.}
Microfinance is a category of financial services that gives small loans to low-income people who may not have access to or be eligible for conventional finance~\cite{yunus2007banker, Econ-MF-Edition2,Kamanza2014CausesOD}. Microfinance has demonstrated potential for poverty reduction, financial inclusion, and economic development~\cite{schreiner2001seven, mersland2010microfinance}. Despite the proven potential, microfinance has experienced several hardships, primarily due to the increase in loan default rates~\cite{nawai2012factors, addae2014causes}.

Although there are many lending strategies and risk control methods in regular finance, they cannot be directly applied to microfinance for the following reasons. First, most existing methods use credit scores to predict the loan default probability when making lending decisions~\cite{ala2016classifiers, shi2019credit, ampountolas2021machine}. However, under-explored (developing) areas without prior loan histories or proper financial systems have insufficient data to establish such credit-scoring procedures. Second, due to the lack of proper state mechanisms, it is difficult for some applicants to provide sufficient information for estimating their credit scores and default probability accurately\footnote{\label{ft:missing-info} For example, it may be costly for some applicants in Africa to obtain proof of residence.}. Third, regular loans are given to individuals with collateral, whereas microfinance often uses group liability, where all individuals in the group are liable if any borrower defaults, to secure repayment~\cite{lehner2009group,kodongo2013individual, haldar2016group}. Group liability can improve the repayment rate by incentivizing members to look after each other, but it has the pitfalls of inducing defaults for borrowers who otherwise have the ability to repay. However, because the approaches for granting regular loans do not sufficiently account for the complexities of group liability, group loans have tended to result in greater default rates~\cite{nandhi2012incidence,allen2016optimal}. Fourth, there is increasing evidence that loan approval algorithms based on black-box machine learning techniques may be biased and discriminatory against minorities~\cite{zliobaite2015survey,corbett2018measure}. Such biases are particularly problematic in fulfilling the objectives of microfinance to provide more opportunities for disadvantaged populations and underdeveloped regions. The complexities of having multiple such populations/regions also imposed additional challenges in allocating microfinance resources to balance different fairness/inclusion objectives. Due to the lack of methodologies that can systematically balance the risks, fairness, and multi-faceted objectives of microfinance, microfinance has relied heavily on the judgment of loan officers. Such operations have sometimes resulted in decisions that let Portfolio at Risk (PAR)\footnote{PAR is defined as the percentage of overdue loans.} exceed a level that is sustainable for continuing microfinance operations\footnote{World Bank suggested 5\% as the upper bound of PAR for sustainability~\cite{ledgerwood2013new}.}~\cite{yimga2016impact,huo2017risk,chikalipah2018credit}.

\textbf{Our focus and contributions.} 
There is an urgent need to modernize microfinance by establishing models and algorithms that account for the aforementioned characteristics and challenges. In this paper, we establish a novel microfinance model and propose an algorithm for learning loan approval policies. We summarize the features of the proposed techniques below. 
\begin{enumerate}[label=(\roman*)]
    \item Our methods can make robust decisions \textit{before} enough loans are given to accurately estimate the default probability and credit scores (\cref{F:changing distribution}) by directly learning the optimal policy parameters without the intermediate step of default probability estimation. 
    \item Our methods degrade more gracefully for increasing levels of missing information in the applications (\cref{F:empty entry comparison}) and exploit the potential of group liability while avoiding its pitfalls (\cref{F:utilities_comparison_10empty_boxplot_group}). The microfinance model accounts for missing information and group liability, and policy learning processes converge to optimal policy parameters in the presence of both. 
    \item Our methods can systematically optimize competing objectives such as risks, socio-economic impacts, and active and passive fairness among different groups (\cref{F:fairness,F:default vs acceptance}). The prioritization among different objectives can be specified in the utility function, and the policy has an interpretable structure that informs which factors contributed positively/negatively to applicant approvals. 
\end{enumerate}
To the best of our knowledge, this paper is the first to use control-theoretic techniques to learn microfinance policy parameters without relying on credit scores directly. Our presentation of microfinance models as a \textit{control problem} opens the door to using modern control and learning theory to modernize microfinance, which in turn helps to achieve 8 of 17 Sustainable Development Goals adopted by the United Nations~\cite{The17Goals}.

\begin{figure}[htbp!]
\centering
\centerline{\includegraphics[width=0.8\columnwidth]{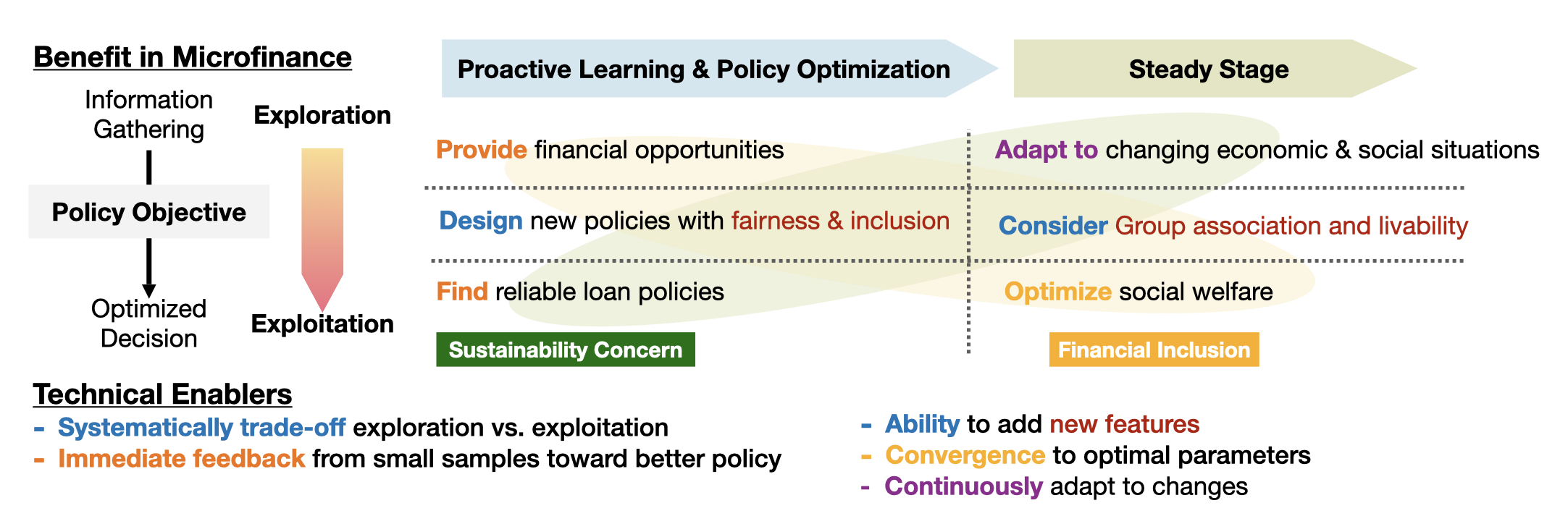}}
\caption{Features of the proposed algorithm and their technical enablers. }
\label{fig:summary}
\end{figure}

\subsection*{Related Works}
\label{SS:Related_Works}

\textbf{Tools developed for regular finance.} The most standard approach to decide regular finance is based on credit scores, which inform the likelihood of each application to default~\cite{klaff2004new,puro2010borrower}. Other approaches consider the loan approval process as a binary classification problem to be solved using machine learning methods such as discriminant analysis~\cite{baesens2003benchmarking}, logistic regression~\cite{ala2016classifiers, vaidya2017predictive}, and neural networks~\cite{abdou2008neural,chen2018loan,condori2021rural}. Multi-layers perceptron neural networks are widely used in automatic credit scoring systems with high accuracy and efficiency~\cite{zhao2015investigation,correa2011genetic}. To achieve a higher prediction accuracy, some studies utilize the random forest approach with feature selection and grid search to reduce the influence of irrelevant and redundant features~\cite{wang2012two, van2016novel}. Other studies adopt the SVM (support vector machine) algorithm to also improve the prediction accuracy with fewer features~\cite{huang2007credit,chen2010combination}. In a different direction, some studies predict the default probability of the applicants utilizing logistic regression and its extensions~\cite{bolton2010logistic,sohn2016technology}. In addition, works such as~\cite{ampountolas2021machine} examine and compare various machine learning algorithms to classify borrowers into various credit categories. However, as mentioned, these approaches mostly focused on offline learning, assuming that accurate homogeneous data is abundantly available. These assumptions may not hold in many under-explored (developing) regions, such as rural Africa, where most of the clients' data is missing. Furthermore, collecting such information might be expensive and the data might be unreliable because of the lack of proper state mechanisms.

\textbf{Fairness in machine learning.} Moreover, machine learning approaches may also introduce fairness issues as it is known to be a controversial topic in the field. Especially for the black-box decision-making tool as it will introduce several explicit and implicit biases to the results~\cite{d2020fairness,corbett2018measure,agarwal2021trade,bantilan2018themis,burrell2016machine}. A few recent works suggested that the decisions made by these black-box processes may have hidden biases and discrimination against under-served populations~\cite{hall2021united, chen2019fairness}. Existing approaches to deal with these biases include modification in the data generation process to ensure the data sets have sufficient diversity~\cite{barbierato2022methodology}, as well as pre-, mid-, and post-processing. For example, \cite{zemel2013learning} introduced a pre-processing approach called LFR (Learned Fair Representations), a discriminative clustering model in which the initial data point is mapped to the distribution in a new input space to conceal any information regarding the data point being a member of a protected subgroup, preserving individual fairness. In another example, \cite{lee2014fairness} proposed the Fairness-Aware BPRMF method as a mid-prepossessing approach by pairing the Bayesian Personalized Ranking Model (BPR) with a combination of matrix factorization (MF). \cite{kim2019multiaccuracy} proposed a post-processing algorithm, MULTIACCURACY-BOOST, that contains an auditor algorithm that iteratively makes mistakes on every sub-population in the black box's classified data until the multi-accuracy constraints of equality are satisfied.

\textbf{Optimization methods.} \label{optimizationfirstord} To find the optimal decision policies, the learning algorithms employ existing optimization techniques in stochastic gradient descent (SGD)~\cite{robbins1951stochastic}, reinforcement learning, and optimal control. For instance, bandit algorithms such as AdaGrad~\cite{duchi2011adaptive}, AdaDelta~\cite{zeiler2012adadelta}, RMSProp~\cite{mukkamala2017variants}, and Adam~\cite{kingma2014adam} are commonly employed. While there is no guarantee for the first-order methods to converge to the optimal solution in general, they will converge to the global optimum solution when the objective function is convex or strongly convex~\cite{zinkevich2003online,hazan2007logarithmic}.

\textbf{Handling missing data} The learning frameworks need to also handle the missing data problem separately. As missing data is a common real-world problem, there is a rich line of literature dealing with it. The classical methods such as mean imputation~\cite{little2002statistical}, expectation maximization~\cite{dempster1977maximum, ghahramani1993supervised, honaker2011amelia}, least squares~\cite{van2011mice, bo2004lsimpute}, K-nearest neighbors~\cite{troyanskaya2001missing}, and regression tree~\cite{burgette2010multiple}, are commonly employed. Recently, an optimization-based approach~\cite{bertsimas2017predictive} has also been considered. 

\textbf{Empirical investigation into microfinance.} Extensive efforts have also been devoted to understanding the economics of microfinance and its sustainability~\cite{Econ-MF-Edition2,CERISE-SPI,MFI-Sustain2010,WB_FSS2018,van2012impact,duvendack2011evidence}. The loan default probability has been determined as one key challenge for microfinance's sustainability. Thus, numerous studies have been carried out to identify the contributing factors to microfinance default behaviors~\cite{Kamanza2014CausesOD,Muthoni2016AssessingIC,Asongo2014TheCO,Kofi2016Causes,dorfleitner2016repayment}. Some features, such as the borrower's gender, education level, family size, residential distance to the institution, lending method, activities financed by the loan, total loan received, and loan monitoring method, among others, have been observed to affect the repayment performance significantly~\cite{field2008repayment, mpogole2012multiple, jote2018determinants,nawai2010determinants}. 

\textbf{Group lending in microfinance.} The studies have also identified that group lending is an essential cornerstone of microfinance. Here, loans are made to small groups or cooperatives that made the members share the liability jointly~\cite{Schurmann_Johnston_2009, besley1995group,ghatak1999group}. This joint-liability model uses social, rather than material, collateral, leveraging peer pressure and community information to overcome asymmetric information in microfinance, leading to better repayment behavior~\cite{PakistanRepayment}. In this model, when one group member defaults, other members jointly bear the cost. Since members of the group are supposed to know each other, this liability structure will help to overcome the information asymmetrical, inherent in lending to poor borrowers, making group-based lending efficient and effective with low transaction costs for the provider~\cite{armendariz2000microfinance,postelnicu2014defining}. A study carried out in Pakistan concluded that with group lending, borrowers are about $60\%$ times as likely to miss a payment in any given month under joint liability relative to individual liability~\cite{PakistanRepayment}. In Africa, there has been also an increasing interest to gear group lending toward traditional group savings structures such as "tontines" in Senegal, "esusu" in Nigeria, "ekub" in Sudan, Eritrea, and Cameroon, or "jangi" in Cameroon~\cite{akanga_reuben_johnson1_2021}. In these Rotating Savings and Credit Associations (ROSCAs), members meet regularly and pay a predetermined sum of money at each meeting~\cite{KIMUYU19991299}. The sum of the payments is then given to a group member (usually the host of the meeting), determined via a lottery in the previous meeting, to pay out the loan of this member~\cite{TontineSenegal_Owen2006}. Among these ROSCAs, the women's association, despite its informal nature, has been one of the most resilient communities where they have survived in many areas where formal microfinancing communities have failed. Because of that, microfinancing through women's associations has recently attracted a lot of interest~\cite{WomenAssoc2002,WomenAssoc2000,abdallahalihal-03375661}.

\section{Problem Statement}
\label{S:Problem_Statement}
In this section, we introduce a novel microfinance model and define the design objectives of the microfinance decision policy.

\subsection{Notation} 
\label{SS:Notation}
We use capital letters for random variables, \eg, $A$, and lowercase letters for their specific realization, \eg, $a$. A square bracket is used to represent the entries of a vector, \eg, $s = \begin{bmatrix} s [1], s[2], \cdots, s[n] \end{bmatrix}^\intercal$, and a regular bracket is used for the input of a function, \eg, $f(x)$. We use $\pro (E)$ to denote the probability of an event $E$ or the density function of a random variable $E$. Lastly, we use $\mathbb{Z}$, $\mathbb{Z}_+$, $\mathbb{R}$, $\mathbb{R}_+$ to denote the sets of integers, non-negative integers, real numbers, and non-negative real numbers.

\subsection{Microfinance model}
\label{SS:Microfinance_model}
A microfinance application can be modeled by application properties, an MFI's decision, and outcomes. The MFI receives applications from individuals or groups of applicants. An application is parameterized by the group size $M$ ($M=1$ for individual applicants), intrinsic features that govern default probability $S \in \mathcal S$, and the MFI's accessible information $\hat{S}\in \hat{\mathcal S}$. We consider the setting that MFI receives applications and makes decisions at each lending period, indexed by $t \in \{1,2,\cdots, T\}$, and continuously learns more optimal policies over the time horizon $T$ as follows. At the beginning of each lending period $t$, the MFI receives $N_t$ financing applications, indexed by $i\in\mathcal{N}_t =  \{1, 2, \cdots, N_t\}$. Application $i$ has group size $m_{i,t} \overset {i.i.d}\sim \pro(M = m_{i,t})$, unobserved underlying features $s_{i,t} \overset {i.i.d}\sim \pro(S)$, and accessible information $\hat s_{i,t} \sim \pro(\hat S \mid S = s_{i,t})$. When some information in $S$ is unavailable, it corresponds to the empty value $\emptyset$ entries in $\hat{S}$. The set of the available information in $\hat S$ is denoted by $U(\hat{S}) =\{j : \hat S[j] \neq \emptyset \}$. The MFI's lending decision is denoted by a random variable $A$:
\begin{align}\label{eq:decision}
    A = \begin{cases}
    1 & \text{for approval}, \\
    0 & \text{for rejection}.
    \end{cases}
\end{align}
The MFI's approval/reject probability, $\pro(A \mid \hat{S}, M)$, for a certain application is based on the lending policy $\pi_Z$, \ie,
\begin{align}
    \label{eq:prob_accptd}
    \mathbb{P}(A \mid \hat{S},M) =  \pi_{Z}(\hat{S}, M, A).
\end{align}
Here, $\pi_{Z}$ is controlled by policy parameter $Z$, defined later in \eqref{eq:approve_parameter}. 

As the amount of loans given out by MFIs is normally small for each individual, we assume the amount of loan and its interest rate are identical among members within the group and, without loss of generality, are set as $1$ and $r$, respectively. Thus, an approved application of group size $M$ receives a loan of size $M$ and must return the principal and interest of $M\cdot(r+1)$ at the end of the lending period, where the loan liability is imposed on the whole group. The outcome of the loan (the ability of the applicant to return) is given by 
\begin{align}
    B = \begin{cases}
        1 & \text{for return},\\
        0 & \text{for default}.
    \end{cases}
\end{align}
We assume $S$ is independently drawn from some underlying population feature distribution $\pro(S)$; $\hat{S}$ is determined based on how features are reflected in the accessible information $\pro(\hat S \mid S)$; and the outcome of the application is governed by $\pro (B \mid S,M) = \pro (B \mid S,\hat{S},M)$, which does not depend on $\hat{S}$ given $S$ and $M$. 

\subsection{Microfinance decision criteria}
\label{SS:Microfinance_decision_criteria}
The MFI uses policy $\pi_{z_t}$ to decide on MFI's action $a_{i,t}\sim \pro(A \mid \hat{S}=\hat{s}_{i,t},\ M=m_{i,t})=\pi_{z_t}(\hat s_{i,t}, m_{i,t}, a_{i,t})$. At the end of the lending period, the MFI observes the loan outcome $b_{i,t} \overset {i.i.d}\sim \pro (B \mid S=s_{i,t},\ M = m_{i,t}) $ and learns (updates) the policy parameter to $z_{t+1}$, which is to be used in the next lending period. Here, $z_t$ represents the learned policy parameter at time $t$ and is updated to minimize the objectives specified by MFIs. Microfinance has multifaceted (non-mutually exclusive) objectives such as financial inclusion, fairness, social and economic impact, and sustainability. These objectives are captured by a utility function $\mathcal{R}(\{\hat{s}_{i,t},m_{i,t},a_{i,t},b_{i,t}\}_{i \in \mathcal{N}_t})$ of all applications $i \in \mathcal{N}_t$.

We defined $V(z_t)$ as the expected utility of the lending period $t$ with policy $\pi_{z_t}$ and control parameter $z_t$, \ie,
\begin{align}
V(z_t) = \mathbb{E}(\mathcal{R}(\{\hat{s}_{i,t},m_{i,t},a_{i,t},b_{i,t}\}_{i \in \mathcal{N}_t})).    
\end{align}
We then consider two types of utility rewards: decomposable and non-decomposable.

\subsubsection{Case 1: Decomposable rewards}\label{para:case1}

Here, we do not consider any external effects of decisions. Therefore, the total reward at time $t$ \textit{can} be decomposed as the sum of individual rewards, \ie, 
\begin{align}\label{eq:decomp_rewards}
\mathcal{R}\left(\{s_{i,t}, m_{i,t}, a_{i,t}, b_{i,t}\}_{i \in \mathcal{N}_t}\right) = \cfrac{1}{{N}_t}\sum_{i = 1}^{{N}_t}R\left(s_{i,t},m_{i,t}, a_{i,t}, b_{i,t}\right),
\end{align}
for some function $R$. For example, we considered the utility function for each application, $R(\hat s, m, a, b)$, in the following form,
\begin{align}\label{eq:utility}
    R(\hat s, m, a, b) = \begin{cases}
    m(r + e); & a = 1, b = 1, \\
    m(-1 + e); & a = 1, b = 0, \\
    0 ; & a = 0.
    \end{cases}
\end{align}
Here, $e \in \mathbb R_+$ is the financial inclusion factor to motivate the MFI toward approving more applications\footnote{MFIs often receive subsidies from international development agencies and governments to help offset high risks of lending without collateral.}.

\subsubsection{Case 2: Non-decomposable rewards}\label{para:case2}

Here, we consider the external effects of MFI's decision. Thus, the total reward \textit{cannot} be decomposed as the sum of individual rewards, but $\mathcal{R}(\{\hat{s}_{i,t},m_{i,t},a_{i,t},b_{i,t}\}_{i \in \mathcal{N}_t})$ will be a function of $R(\hat s, m, a, b)$ and other objectives. This is particularly useful when we design reward functions that account for the social impact on individual decisions such as fairness among different demographics. 

\textbf{Accounting for fairness.} Microfinance decisions should be fair and avoid discrimination against certain populations or regions. Besides the explicit discrimination feature that we could remove directly, we also consider the following two types of implicit fairness. 

\textit{Type 1 (Outcome fairness):} Type 1 fairness actively sets a target approval rate $\Pi^{*}(\xi)$ for the desired approval rate for applications with attribute $\xi$. For example, the loan approval policy of an MFI has a target of at least $\Pi^{*}$ approval rate for female applicants as a criterion for gender equality (\cite{Khaleghi2020}). For type 1 fairness, we want to have 
\begin{align}
    \pro \left(a_{i,t}=1 \mid i\in \mathcal{N}_{t,\xi} \right) \ge \Pi^*(\xi), \forall \xi \in \Xi,
\end{align}
where $\Xi = \{\xi, \xi', \cdots\}$ is the set of attributes for which we have a target approval rate and $\mathcal{N}_{t,\xi} = \{i\in \mathbb{N}:i\in \mathcal{N}_t, \hat{S}_\xi \in \hat{S} \}$ is the set of applications with attribute $\xi$ at time $t$. To achieve type 1 fairness, we can design the reward function as
\begin{align}
    \mathcal{R}\left(\{s_{i,t}, m_{i,t}, a_{i,t}, b_{i,t}\}_{i \in \mathcal{N}_t}\right) = \text{other objectives} - \mathcal{F}_1 \cdot \sum_{\xi \in \Xi} \left(\Pi^{*}(\xi)-\rho_{\xi,t}\right)_{+},
\end{align}
where $\mathcal{F}_1 \in \mathbb{R}_+$ is a weighing factor that reflects the relative importance placed upon type 1 fairness, and $\rho_{\xi, t}$ is the current approval rate for applications with $\xi$, \ie,
\begin{align}
     \rho_{\xi, t} = \cfrac{1}{|\mathcal{N}_{t,\xi}|}\sum_{i\in \mathcal{N}_{t,\xi}}a_{i,t}.
\end{align}
We use the notion $(x)_{+} \text{ to be } \max(0,x)$. 

\textit{Type 2 (Statistical parity):} Type 2 fairness enforces fairness among applications with different attributes, for example, male and female. If we would like to have type 2 fairness among applications with attributes $ \xi$ and $\xi'$, then we should have
\begin{align}
    \pro \left( a_i=1 \mid i\in G_{\xi} \right) \approx \pro \left( a_i=1 \mid i \in G_{\xi'} \right). 
\end{align}
To enforce type 2 fairness, we can adjust our reward as,
\begin{align}
    \mathcal{R}\left(\{s_{i,t}, m_{i,t}, a_{i,t}, b_{i,t}\}_{i \in \mathcal{N}_t}\right) = \text{other objectives} - \mathcal{F}_2 \cdot \lVert \rho_{\xi,t} - \rho_{\xi',t} \rVert,
\end{align}
where a larger value of $\mathcal{F}_2$ indicates more emphasis is put on type 2 fairness.

\subsection{Policy learning objectives} 
\label{SS:Policy_learning_objectives}
The MFIs update the policy parameter $z_t$ to converge to the optimal parameters,
\begin{align} \label{eq:object_function}
    z^* = \argmax_{z} V(z), 
\end{align}
and the cumulative utility converges to the optimal policy fast with low policy exploration cost, quantified by
\begin{align} \label{eq:object_function_convergence}
    V(z^*) -\mathbb{E}\left [  \cfrac{1}{T} \sum_{t = 1}^T V(z_t) \right].
\end{align}
The expectation in $V(z)$ is taken over the probability measure involving group liability and missing information. Thus, by construction, the optimization of \eqref{eq:object_function} and \eqref{eq:object_function_convergence} also accounts for the above-mentioned challenges. Additionally, we take interpretability into the design consideration by imposing structures in policy $\pi_z$ so that the parameter $z$ informs how much each entry of available information $\hat S$ and the group size $M$ have contributed toward approvals or denials. 

\section{Methodology}
\label{S:Methodology}
\subsection{Proposed algorithm}
\label{SS:Proposed_Algorithm}

Here, we will introduce novel learning techniques that can produce optimal and fair microfinance decisions when credit-scoring systems cannot function properly due to the scarcity of prior loan data and the uncertainty of missing data. We propose a lending policy $\pi_z$ and consider the decision policy in the form
\begin{align}
    \pi_z(\hat{s}, m, a) = L(q).
\end{align}
$L(q)$ can be any continuously-differentiable monotonically-increasing function of $q$ that map the domain of $q$ to $(0,1)$. For example, we consider the following choice of $L(q)$:
\begin{align} 
\label{eq:L(q)}
L(q) = \cfrac{2\exp(q)}{1+\exp(q)}-1.
\end{align}
Here, $L(q)$ can be thought of as the activation function of the neural networks. However, the traditional approaches to directly employing the neural networks for lending decisions can aggregate the biases toward the initial choice of the approved applicants because populations who never get loan approval are not contained in the data used to learn the decision policy. Unlike the traditional approach, under our proposed policy, people who are less likely to get approved have a non-zero probability of approval to ensure diversity in the training data. 

Here, $\pi_z$ is parameterized by
\begin{align}
\label{eq:approve_parameter}
&z = [\phi^\intercal, \epsilon^\intercal, \gamma^\intercal]^\intercal \in \mathcal{Z} \subset \mathbb {R}^{2n + |\Xi|}, 
\end{align}
where $n$ is the number of features. Intuitively, $\phi \in R^n$ gives weight to the relative importance of each feature to the decision, $\epsilon \in R^n$ accounts for missing information, and $\gamma \in R^{|\Xi|}$ handles the fairness considerations. We update the policy parameter $z_t$ according to:
\begin{align}
    \label{eq:z_update}
    \hat{z}_{t+1} & = z_t + \alpha_t F_{z_t},\\
    \label{eq:z_proj}
    z_{t+1} & = \proj_{\mathcal{Z}} (\hat{z}_{t+1}).
\end{align}
Readers familiar with learning theory would anticipate that $F_{z_t}$ is the gradient of the objective. This is indeed the case as shown in \cref{lm:F equal grad,lm:F no decomsable}. In the following, we give the explicit form of $F_{z_t}$ that is dependent on the rewards, while the algorithm and theoretical analysis can be found in \cref{SS:Convergence_Analysis}.

For the decomposable rewards, we consider 
\begin{align} \label{eq:q_decomposible}
    q =  \cfrac{1}{n} \sum_{j \in U(\hat{s})}\phi[j]\hat{s}[j]  +\epsilon[j].
\end{align}
In this case, $z = [\phi^\intercal, \epsilon^\intercal, \gamma^\intercal] = [\phi^\intercal, \epsilon^\intercal]$ as $\gamma_{\xi} = 0,  \forall \xi$ and $F_{z_t} = \big[F_{z_t}[1], F_{z_t}[2], \cdots F_{z_t}[2n]\big]^\intercal$ is given by 
\begin{align}\label{eq:gradient_long}
    &\ F_{z_t}[k] 
    = \cfrac{1}{N_t}\sum_{i=1}^{N_t} w_{i,t} [k] \big(R(\hat{s}_{i,t}, m_{i,t}, a_{i,t},  b_{i,t}) - \bar{R}_t\big),\\
    &w_{i,t} [k] = \frac{1}{\pi_{z_t}(\hat{s}_{i,t}, m_{i,t}, a_{i,t})} \cfrac{\partial \pi_{z}(\hat{s}_{i,t}, m_{i,t}, a_{i,t})}{\partial z[k]},\\
    &\bar{R}_t = \cfrac{1}{t-1}\sum_{\tau=1}^{t-1} \cfrac{1}{{N}_\tau}\sum_{i=1}^{{N}_\tau} R(\hat{s}_{i,\tau},m_{i,\tau},a_{i,\tau}, b_{i,\tau}).
\end{align}
Here, $\cfrac{\partial \pi_{z_t}(\hat{s}_{i,t}, m_{i,t}, a_{i,t})}{\partial z[k]} = g(\hat{s},k) \cfrac{dL(q)}{dq}$ is the partial derivative of $\pi_z(\hat{s}_{i,t}, m_{i,t}, a_{i,t})$ with respect to the $k$-th entry of $z$ evaluated at $z_t$, where $g$ is defined to be
\begin{align} \label{eq:phi_partial}
    g(\hat{s},k) = \begin{cases}
        \hat{s}[k]; & k\leq n, k \in U(\hat{s}), \\
        1; & k\geq n+1, k \in U(\hat{s}), \\
        0; & \text{otherwise}.
    \end{cases}
\end{align}

For the non-decomposable rewards, we consider
\begin{align}
\label{eq:predicted_probability}
q = \cfrac{1}{n} \sum_{j \in U(\hat{s})  } \phi[j]\hat{s}[j]  +\epsilon[j] + \sum_{\xi} \gamma_{\xi} \ \mathbbm{1}  \{\hat{s}\in\hat{S}_\xi \},
\end{align}
where $\mathbbm{1} \{\hat{s}\in\hat{S}_\xi \}$ is an fairness indicator function of $(\hat{s}\in\hat{S}_\xi)$. $\hat{S}_\xi$ refers to feature values that may introduce discrimination in microfinance lending decisions, such as gender, race, ethnicity, etc. In this case, $z = [\phi^\intercal, \epsilon^\intercal, \gamma^\intercal]^\intercal$, $\gamma = [\gamma_1, \gamma_2, \gamma_3, \cdots]^\intercal$, and $F_{z_t} = \big[F_{z_t}[1], F_{z_t}[2], \cdots F_{z_t}[2n], \cdots, F_{z_t}[2n+|\Xi|]\big]^\intercal$, where $F_{z_t}$ is given by 
\begin{align}
    \label{eq:fair_gradient_long}
    &F_{z_t}[k] = \left(\mathcal{R}(\{s_{i,t}, m_{i,t},a_{i,t}, b_{i,t}\}_{i \in \mathcal{N}_t}) - \bar{R}_t \right) w_t[k],\\
    &w_t[k] = \sum_{i=1}^{N_t}w_{i,t}[k],\\
    &w_{i,t} [k] = \frac{1}{\pi_{z_t}(\hat{s}_{i,t}, m_{i,t}, a_{i,t})} \cfrac{\partial \pi_{z}(\hat{s}_{i,t}, m_{i,t}, a_{i,t})}{\partial z[k]},\\
    &\bar{R}_t = \frac{1}{t-1}\sum_{\tau=1}^{t-1}R(\{s_{i,t},m_{i,t}, a_{i,t}, b_{i,t}\}_{i \in \mathcal{N}_t}).
\end{align}
Here, $\cfrac{\partial \pi_{z_t}(\hat{s}_{i,t}, m_{i,t}, a_{i,t})}{\partial z[k]} = g(\hat{s},k) \cfrac{dL(q)}{dq}$ with $g$ is defined to be
\begin{align} \label{eq:phi_partial_non_decomposed}
    g(\hat{s},k) = \begin{cases}
        \hat{s}[k]; & k\leq n, k \in U(\hat{s}), \\
        1; &  n + 1 \leq k\leq 2n, k \in U(\hat{s}), \\
        \gamma_{k-2n}; & 2n+1 \leq k \leq 2n + |\Xi|, \\ 
        0; & \text{otherwise}.
    \end{cases}
\end{align}

In both cases, the value of $R(\hat{s}_{i,t}, m_{i,t}, a_{i,t}, b_{i,t}) - \bar{R}_t$ and the update size of $z_t$ will become small when $\bar{R}_t$ is sufficiently close to the sample average. The above procedures are summarized in~\cref{alg:parameters_update}. The value of $\alpha_t > 0$ is the step size to update the parameters at lending period $t$. The choice of $\alpha_t$ is crucial for the convergence and learning speed of the algorithm and is studied theoretically in \cref{crl:step_size} and empirically in \cref{SSS:Step_size}. To make sure that the updated parameters stay in the allowable domain $\mathcal{Z}$, step \eqref{eq:z_proj} projects $\hat{z}_{t+1}$ onto domain $\mathcal{Z}$.

\begin{algorithm}[htbp!]
	\caption{Policy Update}
	\label{alg:parameters_update}
	\begin{algorithmic}
	\STATE Initialized $z_1$
	   \FOR {each lending period $t$}
	        \FOR {each application $i$}
    	        \STATE Generate the decision of application $i$ with: 
        	             \begin{align*}
        	                 a_{i,t} = \begin{cases}
        	                    1; & \text{with probability } \pi_{z_t}(\hat{s}_{i,t},1), \\
        	                    0; & \text{with probability } \pi_{z_t}(\hat{s}_{i,t},0).
        	                \end{cases}
        	             \end{align*}\\
                \STATE Observe outcome $b_{i,t}\in\{0,1\}$.
                \STATE Gain utility $R(\hat{s}_{i,t},a_{i,t},b_{i,t})$.
	        \ENDFOR
	        \STATE Compute $F_{z_t}$ from  \eqref{eq:gradient_long} when choosing general case, or \eqref{eq:fair_gradient_long} when choosing fairness case. 
	        \STATE Update $z_{t+1}$ based on \eqref{eq:z_update} and \eqref{eq:z_proj}.
	   \ENDFOR
	\end{algorithmic}
\end{algorithm}

Our proposed algorithm has the following features: 
\begin{enumerate}[label=(\roman*)]
    \item Financial inclusion and exploration. The stochastic policy provides approval probability $\pro(A|\hat{S})$ rather than a specific lending decision (approve or reject), which ensures sufficient diversity in samples.
    \item Account for missing data. The model contains parameter $\epsilon[j]$ to differentiate zero-value data and empty data.
    \item Fairness consideration for heterogeneous features. We consider three types of fairness: independence, outcome fairness, and statistical parity which may introduce discrimination against applicants of certain types. Our decision framework explicitly models fairness types 2 and 3 by using the parameter $\gamma$. We also consider their solutions for the fair allocation of resources for different groups of applicants.
    \item Each piece of accessible information contributes to the policy. Each feature will either positively or negatively contribute to the probability of approval.
    \item Convergence under mild assumptions. Specifically, we will show that the average utility approaches the optimum in the long run (see \cref{thm:convergence,crl:step_size}).
    \item Parameters learned directly from the algorithm. 
\end{enumerate}

\subsection{Optimality and convergence analysis}
\label{SS:Convergence_Analysis}

This section provides conditions that ensure the proposed algorithm converges to optimal parameters (\cref{thm:convergence}). Along the way, we explain the ideas behind the updating rules \eqref{eq:z_update} in \cref{lm:F equal grad}. we also find an appropriate choice for the step size $\alpha_{t}$ in \cref{crl:step_size}, based on the results of \cref{thm:convergence}.

\textit{Convergence condition.} \cref{alg:parameters_update} converges to the optimal parameters when the following conditions are fulfilled:
\begin{enumerate}[label=(\roman*)]
    \item $L(q)$ is concave over the domain of $z$. \label{condition 1}
    \item The set of admissible policy parameters\label{condition 2} $\mathcal{Z}$ satisfies
        \begin{align}
            \mathbb{E}\left[\|z_{t_1}-z_{t_2}\|^2\right] \leq D^2,\ \forall z_{t_1},z_{t_2} \in \mathcal{Z} \label{eq: assmptn2}.
        \end{align}
    \item The second moment of the stochastic gradient is bounded,\label{condition 3}
        \begin{align}
            \mathbb{E} \left[\|F_{z_t} \|^2 \mid z_t \right] \leq G^2.
        \end{align}
\end{enumerate}
The convergence conditions are formally stated in the theorem below.
\begin{theorem}\label{thm:convergence}
Assuming conditions \ref{condition 1}, \ref{condition 2}, and \ref{condition 3} hold, let $C(T)$ be defined by 
\begin{align} \label{eq: assmptn3}
    C(T) = \sum_{t=1}^T \alpha_t,
\end{align}
and $z^* = (\phi^*,\epsilon^*, \gamma^*)$ be defined in \eqref{eq:object_function}. Then, \cref{alg:parameters_update} gives the following performance:
\begin{align}\label{eq:convergence}
    \mathbb{E}\left[\sum_{t=1}^T\big(V(z^*)-V(z_t)\big)\right] \leq \cfrac{1}{2} \left(\cfrac{1}{\alpha_T} D^2 + G^2 C(T) \right).
\end{align}
\end{theorem}

\noindent \cref{thm:convergence} gives the following relation between step size and convergence speed.  
\begin{corollary} \label{crl:step_size}
When step size is chosen to be $\alpha_t = \cfrac{D}{G\sqrt{t}}$ , we have
\begin{align} \label{eq:step_size_thm}
    \mathbb{E}\left[\cfrac{1}{T}\sum_{t=1}^T \left(V(z^*)-V(z_t)\right)\right] \leq \cfrac{3DG}{2\sqrt{T}}.
\end{align}
\end{corollary}

\cref{thm:convergence} relies on the property that
\begin{align}
\label{eq:gradient_property}
    \mathbb{E}\left[F_{z_t} \mid z_t\right] = \nabla_{z} V(z_t).
\end{align}
\begin{lemma}\label{lm:F equal grad}
When the total reward can be decomposed as the sum of rewards from individual applicants as in \eqref{eq:decomp_rewards}, the updating rules \eqref{eq:z_update} and \eqref{eq:z_proj} given \eqref{eq:gradient_long} - \eqref{eq:phi_partial} satisfy \eqref{eq:gradient_property}.
\end{lemma}
\begin{lemma}\label{lm:F no decomsable}
When the total reward cannot be decomposed as the sum of rewards from individual applicants, the updating rules \eqref{eq:z_update} and \eqref{eq:z_proj} given \eqref{eq:fair_gradient_long} - \eqref{eq:phi_partial_non_decomposed} will also satisfy \eqref{eq:gradient_property}.
\end{lemma}

Here, from an online learning perspective, one can interpret that the algorithm performs a form of stochastic gradient on reward function $V$ in the presence of missing information where, from a control perspective, $V(z)$ can be interpreted as a Lyapunov function. From \cref{thm:convergence}, optimization problem \eqref{eq:object_function} converges to the optimal parameters as $T\rightarrow \infty$ for the concave utility function. When the utility function of interest is not concave, further care might be needed. The \cref{prop: concavity} gives conditions under which the utility function is concave. 
\begin{proposition}
If the approval probability  $L(q)$ is a concave function of $q$, then the objective function $V(z)$ is concave in $z$. 
\label{prop: concavity}
\end{proposition}
\noindent The decision rule \eqref{eq:L(q)} in \cref{S:Problem_Statement} is an example of concave function in $q$ over the positive domain. \cref{prop: concavity} and condition \ref{condition 1} imply the expected rewards derived from concave decision rules such as \eqref{eq:L(q)} are also concave. The proofs of the above theorem, corollary, lemma, and proposition are derived in \cref{S:Proofs}.

\section{Experimental Settings}
\label{S:Experimental_Settings}
We considered the utility function (5) and set $r = 0.35$~\cite{kneiding2008variations}. The simulations were performed on a computer with an Intel Core i7-10875H processor with 32 GB of Random Access Memory (RAM).

\subsection{Generating application data pools}
\label{SS:generating data}

To investigate the empirical behaviors of the proposed methods, we generated synthetic and artificial data pools from more than 30 different distributions where each data pool contains $10^6$ data samples. The data generation process is described below:

\subsubsection{Realistic synthetic dataset from Kiva dataset} 
\label{SSS:kiva data generation}
We augmented a real microfinance loan dataset based on applicants' features and information from the Kiva platform~\cite{hartley2010kiva}\footnote{The original kiva dataset can be accessed through \url{https://stat.duke.edu/datasets/kiva-loans}}. The Kiva dataset contains the features and returns/default information of $N = 3,182$ approved loan applications in 2011. There are 41 features in the raw Kiva dataset, and we considered the features that do not conflict with our assumptions (see \cref{tab:kiva dataset} for the description of the features and our decision). To be specific, we used the following labels: "id", "description.languages", "activity", "sector", "location.country\_code", "location.country", "location.town", "location.geo.level", "partner\_id", "borrowers.gender", "borrowers.pictured", "terms.disbursal\_currency", "terms.loss\_liability.nonpayment", "terms.loss\_liability.currency\_exchange", "posted\_date", "funded\_date", "video.youtube\_id", "lat", and "lon". We used the "posted\_date" and "funded\_date" information to obtain the duration until the loan is funded from the date when the application is posted on the platform. After filtering the features, we first converted all the descriptive features into numerical values based on the default rate of applications. To ease the optimization process, we map the numerical values with a maximum possible value above $4$ to the $[0,4]$ range, utilizing the following process,
\begin{align}
    \label{eq:rescale kiva}
    \centering
    s_{\text{scaled}}[k] = \cfrac{4\cdot s_{\text{kiva}}[k]}{\max(s_{\text{kiva}}[k])},
\end{align}
after shifting features with negative values up, such that the lowest values equal to zero. Here, $\max(s[k])$ refers to the maximum value of the $k$ feature. We then employed the Synthetic Data Vault (SDV) algorithm~\cite{The2016Vault}, available at \url{https://sdv.dev/}, to generate synthetic data from the Kiva dataset.

We augmented the pre-processed data by creating a model using \texttt{TabularPreset()} function from the \texttt{sdv.lite} library. We chose to use the available \texttt{FAST\_ML} preset to fit the data. Then, we generated a pool of $10^6$ synthetic data by inputting all $3,182$ applications data into the \texttt{fit()} function of the model from \texttt{TabularPreset()} and generated the synthetics data using the \texttt{sample()} function of the model. To simulate situations with various different paid rates, we also generated pools of synthetics data with 10\%, 20\%, 30\%, 40\%, and 50\% defaulted rates. Here, we first differentiate the data into two categories, defaulted and paid. We then generated pools of $10^6$ defaulted and paid datasets each by inputting the defaulted and paid data into the \texttt{fit()} function of the model, respectively. The pools of data with desired defaulted rate are then created by randomly drawing applications from both defaulted and paid synthetics data separately.

\subsubsection{Artificial distributions} 
\label{SSS:artificial data}
We generated artificial data pools on $30$ different distribution types. We considered the feature vector $S$ of $100$-features information from each loan application. Each feature contains a non-negative number that represents the applicant's information. Those features could go from personal information, including age range, income level, education level, language skill, etc., to household information such as household type, number of bedrooms, internet accessibility, etc. We considered $26$ distribution types with bounded values of feature information in the range of $[0,4]$ and $4$ distribution types with unbounded values of feature information. From the $26$ bounded distributions, $8$ of them have negative weight for the features. The specific form of the feature information distributions $\pro\left(S\right)_l$ and the return probability $\pro(B=1 \mid \hat{S},M)$ can be seen in \cref{tab:distribution list,tab:distribution list flip,tab:distribution list unbounded}.

Each application $i$ comes with group size $m$. We assume that the group size is uniformly distributed among $\mathcal{G} = \{1,2,\cdots, 100\}$, 
\begin{align}
    \pro(M=m) = \begin{cases} 
        1/|\mathcal{G}| & \text{if } m \in \mathcal{G} \\
        0 & \text{otherwise}.
    \end{cases}
\end{align}
The liability for the loan repayment is imposed on the group by obligating the members of the group to cover the other members who cannot return the loan and its interest, $1+r$. At the end of the lending period, each applicant $j$ in the group $i$ holds $\theta_{i,j}$ unit of money, which is the principal plus gains or loss. We assume that $\theta_{i,j}$ is an \textit{i.i.d.} Gaussian random variable, \ie,
\begin{align}
    \theta_{i,j} \overset{\text{i.i.d.}}{\sim} \mathcal{N}(\mu(\hat{s}_i),\sigma^2).
\end{align}
where $\mathcal{N}(\mu,\sigma^2)$ refers to a Gaussian distribution with mean $\mu$ and standard deviation $\sigma$. Thus, at the end of the lending period, the group $i$ holds $\sum_{j=1}^{m_i} \max\{\theta_{i,j},0\}$ units of money in total. Here, we assumed that the members with $\theta_{i,j} < 0$, who hold debts from elsewhere, do not have the capability to return any money but will not transfer the other debts to the group. Because group $i$ must return $m_i\cdot(1+r)$, the group will default if 
\begin{align}
    \sum_{j=1}^{m_i}\max (\theta_{i,j},0) < m_i\cdot(1+r).
\end{align}
Accordingly, the return probability for group application with $m_i$ members is given by
\begin{align}
    \pro(B=1\mid\hat{S} = \hat{s}_i, M = m_i) = \pro\left(\sum_{j=1}^{m_i}\max (\theta_{i,j},0)\geq m_i\cdot(1+r)\right).
    \label{eq:basic group}
\end{align}
We considered the repayment probability in \cref{tab:distribution list,tab:distribution list flip,tab:distribution list unbounded}, where we modeled $\mu(\hat{s}) = \Delta + \pro(B=1 \mid \hat{S} = \hat{s}, M = m)$ and $\sigma = 0.5$ given $\Delta = 0.5+r \text{ for } r = 0.35$, that helped to center $\mu$ around a reasonable repaid amount.

\subsection{Simulation Setting}
\label{SS:simulation setting}

We performed simulations for lending periods of $t=1$ to $t=500$. At each lending period, we choose $N_{t}=10$ applicants, taken from the pool of applicants generated using methods described in \cref{SS:generating data}. Here, we have the observable information $\hat{s}_{i,t}$.

\subsubsection{Algorithms to compare}
\label{SS:Algorithms to Compare}
We compared our proposed algorithm against several existing algorithms. As the base ideal performance, the best scenario to decide on loan approval is when perfect knowledge of repayment probability is available. Here, for the perfect decision making, \eqref{eq:utility} can be rewritten as
\begin{align}\label{eq:utility_perfect}
    &\ R(\hat S, M, A, B) \nonumber \\
    = &\ \begin{cases}
    (r + e)\cdot m; & \pro(B=1 \mid A=1, S, M=m), \\
    (-1 + e)\cdot m; & \pro(B=0 \mid A=1, S, M=m), \\
    0 ; & A = 0, B = 0,
    \end{cases} 
\end{align}
for $\pro(B=0 \mid A=1, S, M) = 1 - \pro(B=1 \mid A=1, S, M)$. The expectation of the utility given the decision to approve a given application can then be expressed as
\begin{align}
    &\ \mathbb{E}\big[R(\hat{s}_{i,t},m_{i,t},a_{i,t},b_{i,t}) \mid a_{i,t}=1\big] \nonumber \\
    = &\ m_{i,t} \Big((r+1)\pro(b_{i,t}=1 \mid a_{i,t}=1, s_{i,t},m_{i,t}) + e - 1 \Big).
\end{align}
The algorithm will approve an application when
\begin{align}
    \label{eq:perfect_approval_condition}
    & \mathbb{E}\big[R(\hat{s}_{i,t},m_{i,t},a_{i,t},b_{i,t}) \mid a_{i,t}=1\big] && \geq \ \ 0 &\\
    \Longrightarrow \ \ & \pro(b_{i,t}=1 \mid a_{i,t}=1, s_{i,t},m_{i,t}) && \geq \ \ \cfrac{1-e}{1+r}.
\end{align}
Thus, in this perfect information scenario, the decision rule can be written as
\begin{align}\label{eq:decision_perfect}
    A = \begin{cases}
     1 & \text{if }\pro(B=1 \mid A=1, S, M) \geq \cfrac{1-e}{1+r}  \\  
     0 & \text{otherwise}.
    \end{cases}
\end{align}
    
\textbf{Credit score based method.} In real life scenario, we cannot access the actual value for the repayment probability $\pro(B=1 \mid A=1, S)$. Here, the credit score based method is commonly employed to predict the repayment probability. In our study, we implemented credit score based method by finding the best fit model for the data from the previous lending period to predict the repayment probability. Specifically, we fit the data into the first-order Gaussian model,
\begin{align}
    \label{eq:gauss1}
    \hat{\pro}(B=1 \mid A=1, \hat{S}, M) & = a_g \exp\left(-\left(\cfrac{q_g-b_g}{c_g}\right)^2\right), \\
    \intertext{where}
    q_g & = \cfrac{1}{n} \sum_{j\in U(\hat{s})} \hat{s} [j].
\end{align}
Here, $\hat{\pro}(B=1 \mid A=1, \hat{S})$ is the predicted repayment probability parameterized by constants $a_g,\ b_g, \text{ and } c_g$. The constants $a_g,\ b_g, \text{ and } c_g$ are obtained by minimizing the non-linear least square of the difference between the predicted and the actual repayment probability, \ie,
\begin{align}
\centering
    & \{a_g,b_g,c_g\} && \nonumber \\
    = & \argmin_{\{a_g,b_g,c_g\}} \sum && \big\| \hat{\pro}(B=1 \mid A=1, \hat{S}, M) \nonumber \\
    & && - \pro(B=1 \mid A=1,S, M) \big\|_2. 
\end{align}
We can then rewrite the decision rule \eqref{eq:decision_perfect} as
\begin{align}\label{eq:utility_predicted_probability}
    A = \begin{cases}
     1 & \text{if } \hat \pro(B=1 \mid A=1,\hat{S},M) \geq \cfrac{1-e}{1+r}  \\  
     0 & \text{otherwise}.
    \end{cases}
\end{align}
The procedure is executed by employing the \texttt{fit()}~\cite{TheMathWorks2019MATLAB2019b} function in MATLAB. For this approach, we assume that the true value for the repayment probability is available at the end of every lending period. The MFI approves all applicant during the first lending period and then store the features data together with the actual repayment probability at the end of the lending period to be fitted to the model for the next lending period. Then, at the end of each lending period after the first, we revise the model by adding more data points using the current data. We repeat the process until the tenth lending period as we observed that there isn't a significant difference in the performance after the tenth lending period.
    
\textbf{Other algorithms compared.} We also compared the performance of the proposed algorithm against perceptron, random forest, support machine vector (SVM), and logistic regression algorithms. In our study, we compared the proposed algorithm against a single-layer perceptron~\cite{rosenblatt1957perceptron} as a classic learning algorithm for binary classification (\ie, approve or decline applications). We employed the MATLAB~\cite{TheMathWorks2019MATLAB2019b} built-in functions for the random forest, SVM, and logistic regression decision making. To be specific, we utilized \texttt{fitctree()} function as the decision tree classification, and \texttt{fitclinear()} function with \texttt{\upquote{svm}} and \texttt{\upquote{logistic}} as the \texttt{\upquote{Learner}} option for the SVM and logistic regression implementation, respectively. We let both functions optimize its hyper-parameters automatically and all other parameters utilize its default values. The decision models were trained at every loan period using the available feature information $\hat{S}$ and the boolean return/default data collected from previous lending periods. Here, at the first lending period, the MFI will approve all applicants and then store the data at the end of the lending period to be used as the training data for the next lending period. Then, at the end of the next lending period, we add the new data to the training data. To save computational costs from training the model, we perform the training only until the tenth lending period.

\subsubsection{Multi-optimization for improved learning speed}
We initiated the optimization scheme with $10$ initial random points and kept five of the best result for the next lending period. To explore the landscape, we then chose additional five random points for the next iteration and performed the optimization iteration for these $10$ points. To reduce the computational cost of the multi-optimization scheme, we only repeated this procedure for the first $50$ lending period and then only performed the gradient optimization for the best result afterward.

\subsubsection{Step size, learning speed, and convergence}
\label{SSS:Step_size}
The choice of step size $\alpha_t$ is crucial for the convergence and learning speed of the proposed algorithm. We studied the effect of different choices of $\alpha_t$ in our proposed algorithm by varying the constant value of $\frac{D}{G} \in \{0.01, 0.05, 0.1, 0.5, 1, 5, 10, 50, 100, 500, 1000, 5000$, $10000, 50000\}$ from \cref{crl:step_size}. A small step size will converge to non-optimum results, while a large step size will produce an overshoot and hinder convergence as can be seen in \cref{F:step size comparison}. \cref{crl:step_size} suggest decreasing the step size with speed $O(1/\sqrt{t})$ which prevents the algorithm to overshoot when the optimum policy has been found.

\begin{figure}[htbp]
	\centering
	\includegraphics[width=0.45\columnwidth]{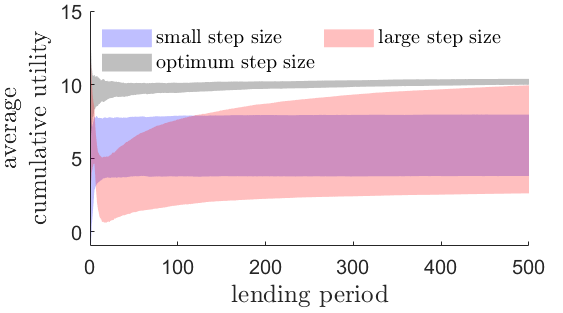}
	\label{F:step size comparison 50 multi}
	\caption{Convergence rate comparison of using different step sizes. For this comparison we consider case A as the form of $L(q)$ with repayment probability distribution type 5 without missing information. We run the simulation 50 times and the shaded regions show the area within one standard deviation of the average cumulative utilities. For this particular case, the optimal step size is achieved when $\alpha_{t} = \cfrac{10}{\sqrt{t}}$.}
	\label{F:step size comparison}
\end{figure}

\section{Performance Comparison}
\label{S:Performance_Comparison}
\begin{figure}[htbp]
    \centering
    \begin{minipage}{0.495\textwidth}
        \subfigure[robustness against missing data]{\includegraphics[width=0.95\textwidth]{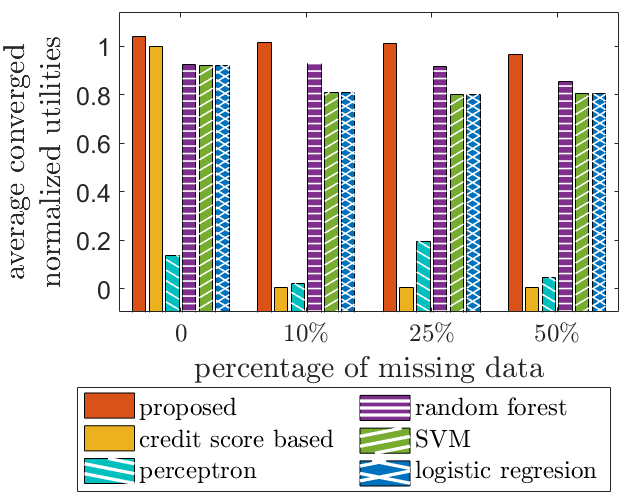}\label{F:empty entry comparison}}\hfill%
    \end{minipage}
    \begin{minipage}{0.495\textwidth}
        \subfigure[fairness]{\includegraphics[width=0.8\textwidth]{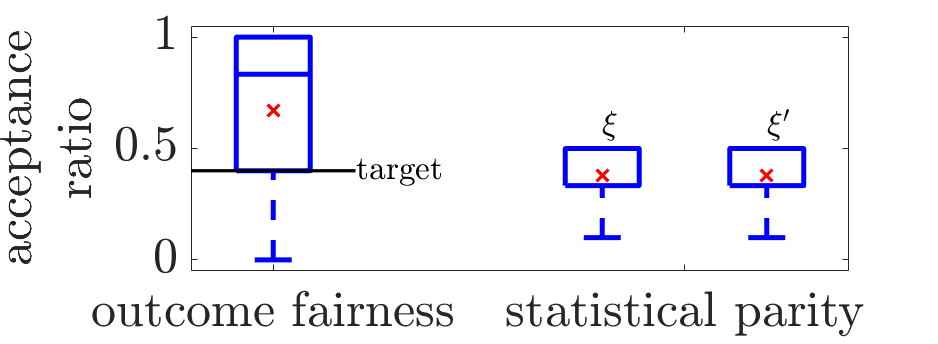}\label{F:fairness}}\hfill%
        \subfigure[default risk vs financial inclusion]{\includegraphics[width=0.975\textwidth]{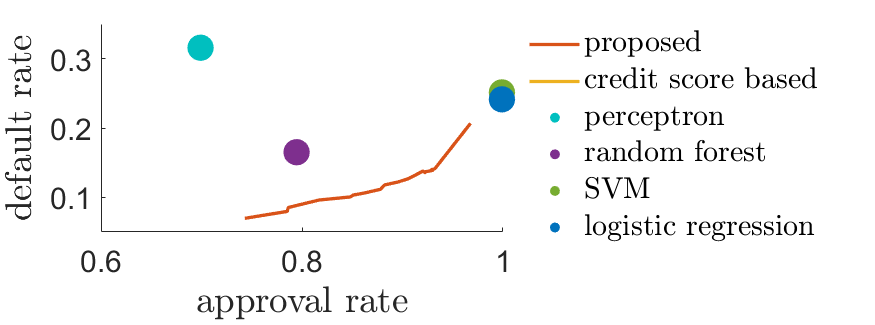}\label{F:default vs acceptance}}%
    \end{minipage}%
    
    \subfigure[artificial data]{\includegraphics[width=0.29\textwidth]{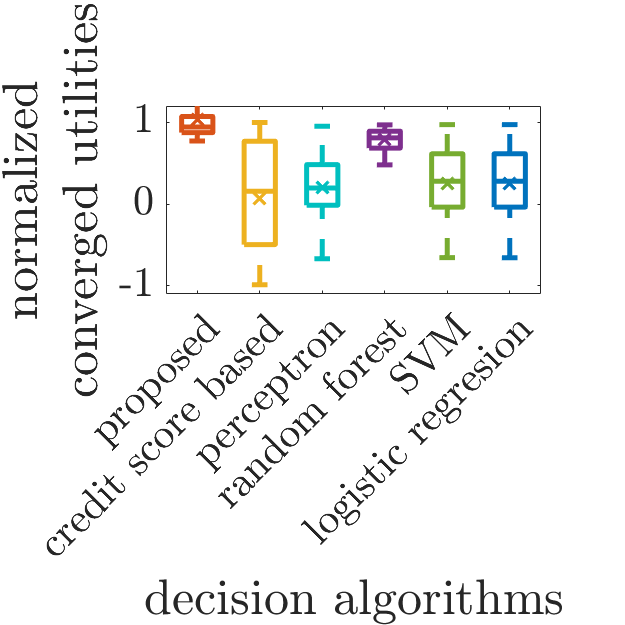}\label{F:utilities_comparison_10empty_boxplot_group}}
    \subfigure[data from kiva.org]{\includegraphics[width=0.29\textwidth]{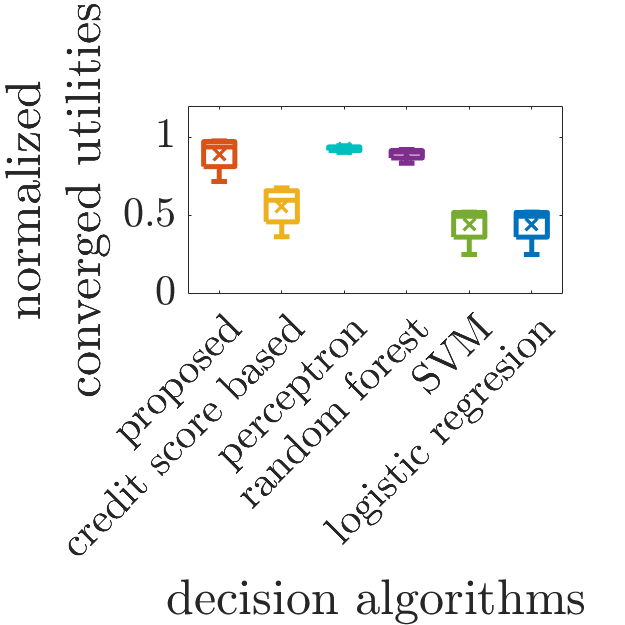}\label{F:utilities_comparison_kiva_boxplot}}
    \subfigure[adaptation]{\includegraphics[width=0.4\textwidth]{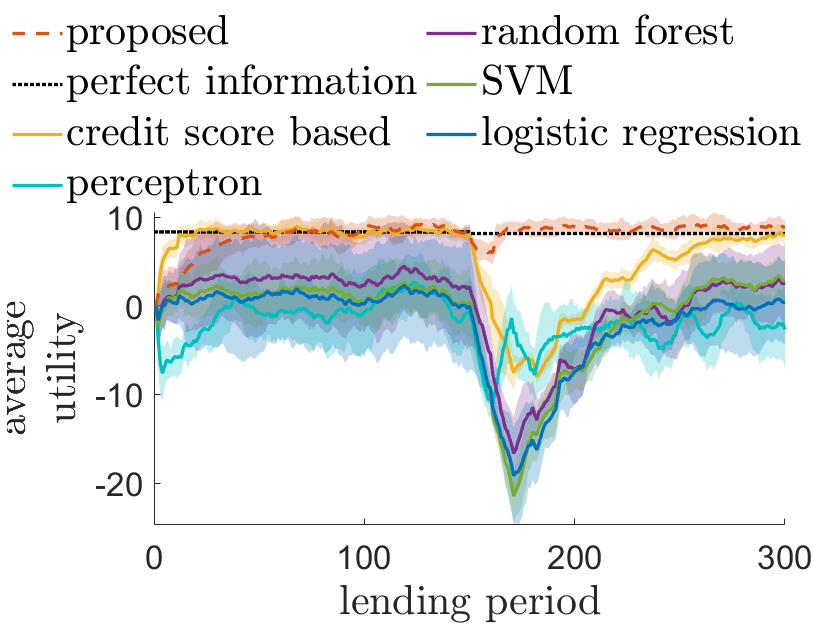}\label{F:changing distribution}}
    
    \caption{(a) Average converged cumulative normalized utilities of 50 simulations are shown for varying rates of missing data. The utility values are normalized by the utility when the exact return/default probability is known. (b) Acceptance ratio statistics of the discriminative features to show fairness types 1 and 2. (c) Tradeoffs between default probability (risk) vs approval rate (financial inclusion). The statistics of the steady state utility for (d) diverse applicant distributions and (e) dataset from kiva. (f) The average utility when the distribution changes at the 150th lending period. In subfigures (b), (d), and (e) the mean values are shown by '$\times$', and the first, second, and third quartiles, as well as the maximum and minimum values, are shown as the boxplot. Except for subfigure (a) and the result from the kiva dataset, the results were generated for $10\%$ of missing information. $20\%$ missing information was considered for the Kiva dataset. Except for (c), we set $e=0$. We consider group liability for all the subfigures.}
    \label{F:contribution}
\end{figure}

\subsection{Robustness against missing data} 
In this study, we set $e = 0$, and performed the simulation $50$ times for each case. We simulated missing information by choosing a missing probability $\pro_{missing}$. Here, the feature information $s_{i,t}$ still had all information, but the corresponding observed entry in $\hat{s}_{i,t}$ was empty with probability $\pro_{missing}$. We varied the missing probability equal to $0$, $10\%$, $25\%$, and $50\%$. \cref{F:empty entry comparison} shows the performance degradation of the algorithms to a varying level of missing feature information. While the performance of all algorithms decreased with the ratio of missing entries, the proposed algorithm degraded more gracefully than others. This is because, unlike the other algorithms, the proposed algorithm is designed to differentiate the empty entry so that the missing information does not affect the decision policy as much (see \cref{S:Problem_Statement}), resulting in an approach that is robust against missing information.

\subsection{Ensuring fairness among different groups} 
To investigate the fairness of our algorithm, we introduced a supposedly discriminative feature to the feature vector. This discriminative feature had three discrete values \{0, 2, 4\} and did not affect the actual repayment probability. The algorithm then ran for some target ratio of the discriminative feature with a value of 0 to investigate the outcome fairness (type 1). To investigate statistical parity (type 2) fairness, we considered the discriminative feature with values 0 and 4. \Cref{F:fairness} shows the box plot for the acceptance rate for the minority group over the course of the simulation. For type 1 fairness, we can see that by setting the target ratio $\Pi(\xi) = 0.4$, the mean of the acceptance rate is above $0.4$. For type 2 fairness, the acceptance rate gap for both groups stays relatively small throughout.

\subsection{Improved tradeoffs between default risk vs. financial inclusion} 
There is a tradeoff between default risk vs. financial inclusion because a higher approval rate comes at the expense of higher default risk. To investigate this property, we varied the loan subsidy level $e$ from $0$ to $1$ with a $0.05$ interval. We then recorded the final approval and default rates of each value of $e$. \cref{F:default vs acceptance} shows such tradeoffs for the algorithms tested. The proposed algorithm allows us to systematically tradeoff default risk and financial inclusion by varying the loan subsidy level $e$. The perceptron, random forest, SVM, and logistic regression algorithms do not have the flexibility to do so because they cannot be optimized for a utility function. The credit score based method is not visible in the current plot range due to large performance degradation in the presence of missing data ($10\%$ missing information). The proposed algorithm achieves a reduced default rate for an identical approval rate.

\subsection{Ability to deal with diverse microfinance distributions} 
We examined the performance of the compared algorithms using the dataset generated from $30$ different distributions and the Kiva dataset as mentioned previously. We ran the simulation $50$ times and took the average of the converged utility values $V(z^*)$ of each distribution type. We then normalized the utilities as follows, 
\begin{align}
    \{\tilde{\mathcal{V}}(z^*)\} = 2\left(\cfrac{\{\mathcal{V}(z^*)\}-\min \{\mathcal{V}(z^*)\}}{V_{perfect}}\right)-1,
\end{align}
where $V_{perfect}$ is the utility when perfect knowledge of repayment probability is available, $\{\mathcal{V}(z^*)\}$ is a set of converged utilities of all compared algorithms, and $\{\tilde{\mathcal{V}}(z^*)\}$ is a set of normalized converged utilities of all compared algorithms. Here, we considered the cases with $10\%$ and $20\%$ missing information on the artificial and synthetic datasets, respectively. \cref{F:utilities_comparison_10empty_boxplot_group,F:utilities_comparison_kiva_boxplot} capture the statistic of the normalized steady-state utilities. On average, our algorithm converges to a higher utility than the other algorithms, suggesting that our algorithm can learn a more optimal policy even with incomplete information, achieving our first design goal. This result also shows that the proposed approach can better seize and learn the influence caused by different group sizes.

\subsection{Adaptation to changes} 
We studied the adaptability of the proposed algorithm to the dynamics in social and economic conditions by changing the distribution of the dataset in the middle of the simulation. To highlight the behavior of the algorithms in adapting to the changes, we chose the distributions that would have similar repayment probability distributions but need to have opposite feature weights. We retrained the credit score based method, random forest, SVM, and logistic regression algorithms at the end of each lending period. \Cref{F:changing distribution} shows the performance when application distribution changes at the 150\textit{th} lending period, without the knowledge of if and when the distribution has changed. As can be seen, the proposed algorithm recovers faster than the other algorithms as it uses immediate feedback from the latest samples to perform quick adaptation. In contrast, the other algorithms, which are primarily designed for offline use, end up putting the same emphasis on data that contains the samples from both before and after the changes.

\section{Conclusion}
\label{S:Conclusion}
In this work, we presented a novel control-theoretic model for microfinance lending strategy. The model solves three main challenges in microfinance: a) the insufficient past data problem, b) the missing applicants' information problem, and c) the group liability structure. Extensive empirical results from numerous artificial and synthetic datasets showed several notable performances upon benchmark models, such as robustness against missing data, adaption to changes, and default risk vs. financial inclusion tradeoff. In addition, we proposed several penalty methods for different fairness scenarios to avoid discrimination in decisions. We hope our model will be useful for achieving the United Nation's Sustainable Development Goals and could help more people in underdeveloped regions to have better life.

\bibliographystyle{ieeetr}
\bibliography{Sections/references}

\appendix
\crefalias{section}{appsec}

\section{Proofs}
\label{S:Proofs}
\textit{\textbf{Proof of Proposition \ref{prop: concavity}}} 
\begin{proof}[\unskip\nopunct]
For the decomposable case, we have 
\begin{align}
\centering
    V(z) & = \mathbb{E} \left[R(\hat{S}, M, A, B)\right] \\
    & =\mathbb{E}\left[\sum_{a=0, 1} \mathbb{E}[R(\hat{S},M, A, B) \mid A=a] \pi_z(\hat{S},M,a)\right] \\
    & = \mathbb{E}\left[ \mathbb{E}[R(\hat{S}, M, A, B) \mid A=1] \pi_z(\hat{S}, M, 1) \right] \label{dagger} \\
    & = \mathbb{E} \left[ \mathbb{E} [R(\hat{S}, M, A, B) \mid A=1] L(q) \right] \label{star},
\end{align}
where \eqref{dagger} is because $\mathbb{E}[R(\hat{S},M, A, B) \mid A=0] = 0$. Since $\mathbb{E}[R(\hat{S},M, A, B \mid A=1]$ is independent from $z$, we can see from \eqref{star} that $V(z)$ is a linear combination of $L(q)$ and from \eqref{eq:q_decomposible}, $q$ is a linear function of $z$. Therefore, if $q$ is a concave function of $z$, then $L(q)$ and $V(z)$ are concave functions of $z$.
\end{proof}
\begin{proof}[\unskip\nopunct]
In the non-decomposable case, we have 
\begin{align}
\centering
    V(z) & = \mathbb{E} \left[R(\hat{S}, M, A, B)\right] \\
    & = \mathbb{E}\left[\sum_{a=0, 1} \mathbb{E}[R(\hat{S},M, A, B) \mid A=a] \pi_z(\hat{S},M,a)\right] \\
    & = \mathbb{E}\Big[ \mathbb{E}[R(\hat{S}, M, A, B) \mid A=1] \pi_z(\hat{S}, M, 1) + \mathbb{E}[R(\hat{S}, M, A, B) \mid A = 0](1-\pi_z(\hat{S}, M, 1))\Big] \\
    & = \mathbb{E} \bigg[ \Big(\mathbb{E} [R(\hat{S}, M, A, B) \mid A=1] - \mathbb{E} [R(\hat{S}, M, A, B) \mid A=0]\Big) L(q) + \mathbb{E} [R(\hat{S}, M, A, B) \mid A=0] \bigg].
\end{align}
Since $\mathbb{E} [R(\hat{S}, M, A, B) \mid A=0]$ and $\mathbb{E} [R(\hat{S}, M, A, B) \mid A=1]$ are independent from $z$, we can see that $V(z)$ is a linear combination of $L(q)$. Therefore, $V(z)$ is concave with respect to $z$ when $L(q)$ is concave.  
\end{proof}

\begin{proof}[\textbf{Proof of \cref{lm:F equal grad}}]
Let $\mathcal{N} = \{1, 2, \hdots, N\}$ be the index of the applications in one lending period. Then, let $\mathbf{\hat{s}} = \{\hat{s}_1, \hat{s}_2, \hdots, \hat{s}_N\},\ \mathbf{m}=\{m_1,m_2,\hdots,m_N\},\ \mathbf{a}=\{a_1, a_2, \hdots, a_N\},\text{ and } \mathbf{b}=\{b_1,b_2,\hdots,b_N\}$ be the vectors containing the accessible information, group size, lending decision, and outcomes of all applicants, respectively. For all $k \in \{1, \ldots, 2N\}$, we have 
\begin{align}
\centering
    \mathbb{E} \Big[F_{z}[k] \mid z \Big] & = \mathbb{E} \Bigg[ \left(\mathcal{R}(\mathbf{\hat{s}},\mathbf{m},\mathbf{a},\mathbf{b})-\bar{R}\right) \sum_{i-1}^{N} \cfrac{1}{\pi_z (\hat{s}_{i},m_{i},a_{i})} \cfrac{\partial \pi_z (\hat{s}_{i},m_{i},a_{i})}{\partial z[k]} \mid z \Bigg]\\
    & = \mathbb{E} \left[\Big(\mathcal{R}(\mathbf{\hat{s}},\mathbf{m},\mathbf{a},\mathbf{b})-\bar{R}\Big) \left(\cfrac{\cfrac{\partial}{\partial z[k]} \pi_z\left(\hat{S},M,A\right)}{\pi_z\left(\hat{S},M,A\right)}\right)\right] \label{eqn:35} \\
    & = \mathbb{E} \left[\rule{0cm}{1cm}\right. \sum_{\mathbf{a}} \left(\rule{0cm}{0.95cm}\right. \left(\rule{0cm}{0.8cm}\right. \cfrac{\pi_z\left(\hat{S},m,a\right)}{\pi_z\left(\hat{S},m,a\right)} \left.\rule{0cm}{0.8cm}\right) \mathbb{E} \Big[\mathcal{R}(\mathbf{\hat{s}},\mathbf{m},\mathbf{a},\mathbf{b}) - \bar{R} \mid A,M,\hat{S} \Big] \left(\cfrac{\partial}{\partial z[k]}\ \pi_z\left(\hat{S},M,A\right) \right) \left.\rule{0cm}{0.95cm}\right) \left.\rule{0cm}{1cm}\right] \label{eqn: 36} \\
    & = \mathbb{E} \left[\rule{0cm}{0.75cm}\right. \mathbb{E} \Bigg[ \sum_{\mathbf{a}} \left( \mathcal{R}(\mathbf{\hat{s}},\mathbf{m},\mathbf{a},\mathbf{b})-\bar{R}\right) \cfrac{\partial}{\partial z[k]}\ \pi_z\left(A,M, \hat{S}\right) \mid A,M,\hat{S} \Bigg]\left.\rule{0cm}{0.75cm}\right] \label{eqn: 37} \\
    & = \cfrac{\partial}{\partial z[k]}\ \mathbb{E} \left[\rule{0cm}{0.6cm}\right. \mathbb{E} \Big[ \sum_{\mathbf{a}} \left( \mathcal{R}(\mathbf{\hat{s}},\mathbf{m},\mathbf{a},\mathbf{b}) - \bar{R} \right)\pi_z\left(A,M,\hat{S}\right) \mid A,M,\hat{S} \Big] \left.\rule{0cm}{0.6cm}\right] \label{eqn: 38} \\
    & = \cfrac{\partial}{\partial z[k]}\ \mathbb{E}\big[\mathcal{R}(\mathbf{\hat{s}},\mathbf{m},\mathbf{a},\mathbf{b})\big] \label{eqn:39} \\
    & = \cfrac{\partial V(z)}{\partial z[k]}.
\end{align}
Equality \eqref{eqn: 36} is by law of total probability, \eqref{eqn: 37} is because of the fact that $\bar{R}$ is independent of $z$ while $\sum_{A\in\{0,1\}} \cfrac{\partial}{\partial z[k]} \pi_z(\hat{S}, A) = 0$, and \eqref{eqn: 38} is by linearity of expectation. Here, the expectations in \eqref{eqn:35} and \eqref{eqn:39} are taken over $\pro(S,\hat{S},A,B)$. The first expectations in \eqref{eqn: 36}, \eqref{eqn: 37}, and \eqref{eqn: 38} are taken over $\pro(S,\hat{S})$ while the second expectations are taken over $\pro(B \mid A, \hat{S})$.
\end{proof}

\begin{proof}[\textbf{Proof of \cref{lm:F no decomsable}}]
Using the same notation as in the Proof of \cref{lm:F equal grad}, we have,
\begin{align}
\centering
    \mathbb{E}\left[F_z[k] \mid z\right] & = \mathbb{E}\Bigg[\Big(\mathcal{R}(\mathbf{\hat{s}},\mathbf{m},\mathbf{a},\mathbf{b})-\bar{R}\Big) \sum_{i=1}^{N} \cfrac{1}{\pi_z(s_i,m_i,a_i)} \cfrac{\partial \pi_z(s_i,m_i,a_i)}{\partial z[k]}\Bigg] \label{prf_lm_4.4_ln1}\\
    & = \mathbb{E}\Bigg[\sum_{\mathbf{a}} \mathbb{E}\Big[\mathcal{R}(\mathbf{\hat{s}},\mathbf{m},\mathbf{a},\mathbf{b})-\bar{R} \mid \hat{\mathbf{s}},\mathbf{m},\mathbf{a},\mathbf{b} \Big] \sum_{i=1}^{N} \cfrac{\pi_z(s_i,m_i,a_i)}{\pi_z(s_i,m_i,a_i)} \cfrac{\partial \pi_z(s_i,m_i,a_i)}{\partial z[k]}\Bigg] \label{prf_lm_4.4_ln2}\\
    & = \mathbb{E}\Bigg[\sum_{\mathbf{a}} \mathbb{E}\Big[\mathcal{R}(\mathbf{\hat{s}},\mathbf{m},\mathbf{a},\mathbf{b})-\bar{R} \mid \hat{\mathbf{s}},\mathbf{m},\mathbf{a},\mathbf{b} \Big] \sum_{i=1}^{N} \cfrac{\partial \pi_z(s_i,m_i,a_i)}{\partial z[k]}\Bigg] \\
    & = \mathbb{E}\Bigg[\sum_{\mathbf{a}} \mathbb{E}\Big[\mathcal{R}(\mathbf{\hat{s}},\mathbf{m},\mathbf{a},\mathbf{b}) \mid \hat{\mathbf{s}},\mathbf{m},\mathbf{a},\mathbf{b} \Big] \sum_{i=1}^{N} \cfrac{\partial \pi_z(s_i,m_i,a_i)}{\partial z[k]}\Bigg] \label{prf_lm_4.4_ln4}\\
    & = \cfrac{\partial}{\partial z[k]}\ \mathbb{E}\Bigg[\sum_{\mathbf{a}} \mathbb{E}\Big[\mathcal{R}(\mathbf{\hat{s}},\mathbf{m},\mathbf{a},\mathbf{b}) \mid \hat{\mathbf{s}},\mathbf{m},\mathbf{a},\mathbf{b} \Big] \sum_{i=1}^{N} \pi_z(s_i,m_i,a_i) \Bigg] \\
    & = \cfrac{\partial}{\partial z[k]}\ \mathbb{E}\Big[\mathcal{R}(\mathbf{\hat{s}},\mathbf{m},\mathbf{a},\mathbf{b}) \Big] \label{prf_lm_4.4_ln6}\\
    & = \cfrac{\partial V(z)}{\partial z[k]}. \label{prf_lm_4.4_ln7}
\end{align}
Here, \eqref{prf_lm_4.4_ln1} is the expectation taken over $\mathbf{\hat{s}},\mathbf{m},\mathbf{a},\text{ and }\mathbf{b}$; the first expectation in \eqref{prf_lm_4.4_ln2} is taken over $\mathbf{\hat{s}} \text{ and } \mathbf{m}$ while the second expectation is taken over $\mathbf{b}$; \eqref{prf_lm_4.4_ln6} holds because $\sum_{a_i\in \{0,1\}} \pi_z(s_i,m_i,a_i) = 1$; \eqref{prf_lm_4.4_ln7} is due to the definition of $V(z)$.
\end{proof}

\noindent\textbf{Gradient inequality lemma.} To bound the regret of a telescoping sum, we present the following:
\begin{lemma} \label{lm:grad to norm}
The following condition holds:
\begin{flalign} \label{eq:grad to norm}
    \mathbb{E}\Big[2 \nabla_{z}V(z_t) \cdot (z_t-z^*) \mid z_t\Big] \leq \cfrac{1}{\alpha_t}\left(\mathbb{E}\Big[\|z_t-z^*\|^2 - \|z_{t+1}-z^*\|^2 \mid z_t\Big]\right) + \alpha_t G^2, 
\end{flalign}
where $V$ is defined in \cref{S:Problem_Statement}.
\end{lemma}
\begin{proof}
Let us consider $\mathbb{E} \Big[ \|z_{t+1} - z^* \|^2 \mid z_t \Big]$, where
\begin{align}
\centering
    \mathbb{E} \Big[ \|z_{t+1} - z^* \|^2 \mid z_t \Big] & = \mathbb{E} \Big[ \|\proj_Z (z_t + \alpha_tF_{z_t}) - z^* \|^2 \mid z_t \Big] \\
    & \leq \mathbb{E} \Big[ \|z_{t} + \alpha_t F_{z_t} - z^* \|^2 \mid z_t \Big] \\
    & \leq \mathbb{E} \Big[ \|z_t - z^* \|^2 + \alpha_t^2 \| F_{z_t} \|^2 - 2\alpha_t F_{z_t} \cdot (z^*-z_t) \mid z_t \Big] \\
    & = \mathbb{E} \Big[ \|z_t - z^* \|^2 + \alpha_t^2 \| F_{z_t} \|^2 \mid z_t \Big] - \mathbb{E} \Big[2\alpha_t F_{z_t} \cdot (z^*-z_t) \mid z_t \Big]\\
    & = \|z_t-z^*\|^2 + \alpha_t^2 \mathbb{E} \Big[\|F_{z_t} \|^2 \mid z_t \Big] - 2\alpha_t \mathbb{E} \Big[F_{z_t} \cdot  (z^*-z_t) \mid z_t \Big] \\
    & = \|z_t-z^*\|^2 + \alpha_t^2 \mathbb{E} \Big[\|F_{z_t} \|^2 \mid z_t \Big] - 2\alpha_t \nabla_{z_t}V(z_t)\cdot(z^*-z_t), \label{eq:lm4thm1} \\
    \intertext{therefore,} 
   2 \nabla_{z_t}V(z_t)\cdot(z^*-z_t) & \leq \cfrac{-\mathbb{E} \Big[ \|z_{t+1} - z^* \|^2 \mid z_t \Big] + \|z_t-z^*\|^2 }{\alpha_t} + \cfrac{\alpha_t^2 \mathbb{E} \Big[\|F_{z_t} \|^2 \mid z_t \Big]}{\alpha_t} \\
    & \leq \cfrac{-\mathbb{E} \Big[ \|z_{t+1} - z^* \|^2 \mid z_t \Big] + \|z_t-z^*\|^2}{\alpha_t} + \alpha_t G^2, \label{eq:lm4thm2} \\
    \intertext{thus,} 
    \mathbb{E} \Big[2 \nabla_{z_t}V(z_t)\cdot(z^*-z_t) \mid z_t \Big] & \leq \cfrac{-\mathbb{E} \Big[ \|z_{t+1} - z^* \|^2 \mid z_t \Big] + \|z_t-z^*\|^2}{\alpha_t} + \alpha_t G^2. 
\end{align}
Here, \eqref{eq:lm4thm1} is due to \cref{lm:F equal grad} and \eqref{eq:lm4thm2} is due to the assumption $\mathbb{E}\left[\|F_{z_t}\|^2 \mid z_t \right] \leq G^2$. Thus, we can then recover \eqref{eq:grad to norm}.
\end{proof}

\begin{proof}[\textbf{Proof of \cref{thm:convergence}}]
Given \cref{lm:F equal grad,lm:grad to norm}, we can proof \cref{thm:convergence}:
    \begin{align}
    \centering
        \mathbb{E}\left[2\sum_{t=1}^T\big(V(z^*)-V(z_t)\big)\right] & \leq \mathbb{E}\left[2\sum_{t=1}^T\big(\nabla_{z}V(z_t) \cdot (z^*-z_t)\big)\right] \label{eq:cnvrge_prf_0}\\
        & \leq \mathbb{E}\left[2\sum_{t=1}^T \mathbb{E}\bigg[\nabla_{z}V(z_t) \cdot (z^*-z_t) \mid z_t\bigg]\right] \label{eq:cnvrge_prf_1} \\
        & \leq \mathbb{E} \Bigg[\sum_{t=1}^T \cfrac{\mathbb{E}\left[\|z_{t+1}-z^*\|^2 + \|z_t-z^*\|^2 \mid z_t\right]}{\alpha_t} +  G^2 \sum_{t=1}^T \alpha_t \Bigg] \label{eq:cnvrge_prf_2} \\
        & \leq \sum_{t=1}^T \mathbb{E}\left[\|z_t-z^*\|^2\right] \left(\cfrac{1}{\alpha_t} - \cfrac{1}{\alpha_{t-1}}\right) + G^2 \sum_{t=1}^T \alpha_t \label{eq:cnvrge_prf_3} \\
        & \leq D^2 \sum_{t=1}^T \left(\cfrac{1}{\alpha_t} - \cfrac{1}{\alpha_{t-1}}\right) + G^2 \sum_{t=1}^T \alpha_t \label{eq:cnvrge_prf_4} \\
        & \leq \cfrac{1}{\alpha_T}D^2 + G^2 C(T). \label{eq:cnvrge_prf_5} 
    \end{align}
    Here, \eqref{eq:cnvrge_prf_0} is due to the concavity of $V(z_t)$; \eqref{eq:cnvrge_prf_1} is due to the law of total expectation; \eqref{eq:cnvrge_prf_2} is due to \cref{lm:grad to norm}; \eqref{eq:cnvrge_prf_3} is due to the law of total expectation; \eqref{eq:cnvrge_prf_4} is due to assumption \eqref{eq: assmptn2}; and \eqref{eq:cnvrge_prf_5} is due to assumption \eqref{eq: assmptn3}.
\end{proof}

\begin{proof}[\textbf{Proof of \cref{crl:step_size}}]
We have
\begin{align}
\centering
    \mathbb{E}\left[\sum_{t=1}^T \left(V(z^*)-V(z_t)\right)\right] & \leq \cfrac{1}{2} \left(D^2 \sum_{t=1}^T \left(\cfrac{1}{\alpha_t} - \cfrac{1}{\alpha_{t-1}}\right) + G^2 \sum_{t=1}^T \alpha_t \right) \label{eq:cr_prf_1} \\
    & \leq \cfrac{1}{2} \left(D^2 \cfrac{G\sqrt{T}}{D} + G^2 \cfrac{D}{G} 2\sqrt{T} \right) \label{eq:cr_prf_2} \\
    & = \cfrac{1}{2} \left(DG\sqrt{T} + 2GD\sqrt{T} \right) \\
    & = \cfrac{3}{2}DG\sqrt{T}, \label{eq:cr_prf_3}
\end{align}
where \eqref{eq:cr_prf_1} is due to \eqref{eq:cnvrge_prf_4}; \eqref{eq:cr_prf_2} is due to $\sum_{t=1}^T \cfrac{1}{\sqrt{t}} \leq 2\sqrt{T}$. We recover \eqref{eq:step_size_thm} by dividing \eqref{eq:cr_prf_3} by $T$.
\end{proof}

\section{Tables of Data}
\label{S:Tables_of_Data}
Here, we present the description of the features in the kiva data together with our decision in \cref{tab:kiva dataset}. We also present 30 different distributions to generate the artificial data pools. The 18 bounded distribution with positive weighting features can be seen in \cref{tab:distribution list}, the 8 distributions with negative weight for the features are shown in \cref{tab:distribution list flip}, and \cref{tab:distribution list unbounded} shows the 4 unbounded distributions.

\begin{table*}[htbp]
\centering
\scriptsize
\begin{tabular}{l l l l}\toprule
Feature Number & Feature Name &Feature Description & Data Preprocessing Decision \\\midrule
1 & id & loan ID & used \\
2 & description.languages & language of loan description & used \\
3 & funded\_amount & amount of loan has been collected by lenders & not available in the beginning \\
4 & paid\_amount & amount of the loan which has been paid off & not available in the beginning \\
5 & activity & activity for requested loan & used \\
6 & Sector & sector for requested loan & used \\
7 & location.country\_code & country code & used \\
8 & location.country & country name & used \\
9 & location.town & town name & used \\
10 & location.geo.level & latitude and longitude indicator for country or town & used \\
11 & partner\_id & partner ID & used \\
12 & borrowers.first\_name & first name of borrower & hard to convert \\
13 & borrowers.last\_name & last name of borrower & hard to convert \\
14 & borrowers.gender & gender of borrower & used \\
15 & borrowers.pictured & borrowers' picture availability & used \\
16 & terms.disbursal\_amount & distributed amount in the local currency & conflict with our assumption \\
17 & terms.disbursal\_currency & distributed currency & used \\
18 & terms.disbursal\_date & distributed date & conflict with our assumption  \\
19 & paid\_date & Fully paid date & not available in the beginning\\
20 & defaulted\_date & defaulted date & not available in the beginning\\
21 & terms.loan\_amount & the amount of money distributed & conflict with our assumption \\
22 & terms.loss\_liability.nonpayment & who is liable for non repayment & used \\
23 & terms.loss\_liability.currency\_exchange & who is liable for currency exchange loss & used \\
24 & posted\_date &loan posted date on Kiva & \multirow{2}{*}{to calculate the duration: used} \\
25 & funded\_date & fully funded date & \\
26 & journal\_total.entries & number of updates by borrower & not available in the beginning \\
27 & terms.local\_payments.due\_date & payment is due date to the field partner & conflict with our assumption \\
28 & terms.local\_payments.amount & amount due to the field partner & conflict with our assumption \\
29 & terms.scheduled\_payments.due\_date & scheduled payment due date & conflict with our assumption \\
30 & terms.scheduled\_payments.amount & scheduled payment due amount & conflict with our assumption \\
31 & delinquent & whether has become delinquent & not available in the beginning \\
32 & video.youtube\_id & youtube id if provide a video & used \\
33 & basket\_amount & amount of loan saved but not confirmed & conflict with our assumption \\
34 & amount & payment amount in US dollars & conflict with our assumption \\
35 & payment\_id & payment ID & conflict with our assumption \\
36 & local\_amount & payment amount in local currency & conflict with our assumption \\
37 & processed\_date & processed date & conflict with our assumption \\
38 & rounded\_local\_amount & rounded local payment amount & conflict with our assumption \\
39 & settlement\_date & payment settlement date & conflict with our assumption \\
40 & lat & latitude of loan location & used \\
41 & lon & longitude of loan location & used \\
42 & status & paid or defaulted status & label \\
\bottomrule
\end{tabular}
\caption{Description of the features and data pre-processing decision of the Kiva dataset. In our study, we only considered the features with the "\textit{used}" label.}\label{tab:kiva dataset}
\end{table*}

\begin{table*}[htbp]
    \centering
    \begin{tabular}{|l||c|c|c|}\hline
        \backslashbox{repayment \\ probability, \\ $\pro\left(B = 1\mid \hat{S},M\right)$}{features \\ distribution, \\ $\pro \left(S\right)$} & 
        \begin{minipage}{.19\textwidth}
          \includegraphics[width=\linewidth]{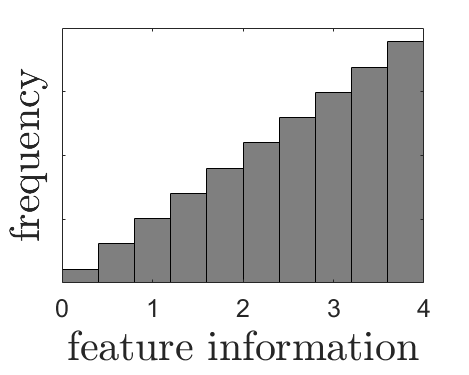}
        \end{minipage} & 
        \begin{minipage}{.19\textwidth}
          \includegraphics[width=\linewidth]{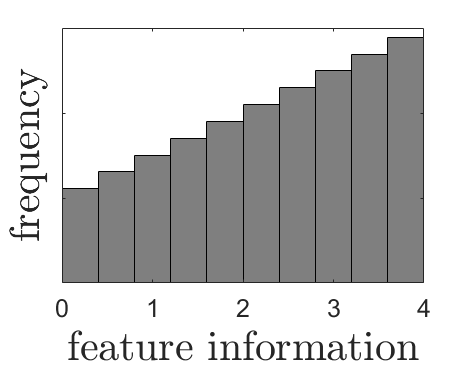}
        \end{minipage} & 
        \begin{minipage}{.19\textwidth}
          \includegraphics[width=\linewidth]{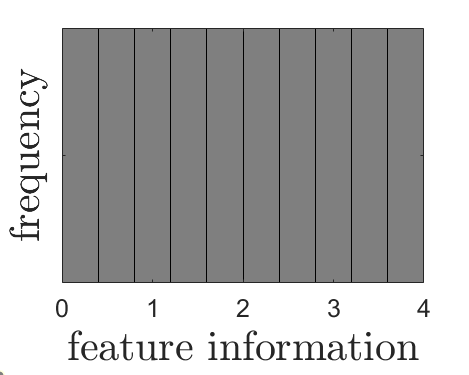}
        \end{minipage} \\\hline\hline
        \begin{tabular}{@{}l@{}} $q_B = \cfrac{1}{n} \sum_{j=1}^n s[j]$ \\ $\pro\left(B = 1\mid \hat{S},M\right) = \cfrac{1}{4}q_B$ \end{tabular} & 
        \begin{minipage}{.19\textwidth}
          \includegraphics[width=\linewidth]{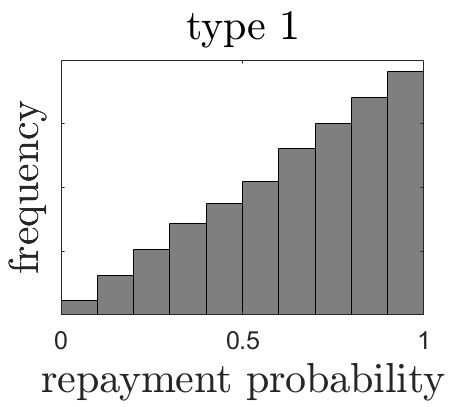}
        \end{minipage} & 
        \begin{minipage}{.19\textwidth}
          \includegraphics[width=\linewidth]{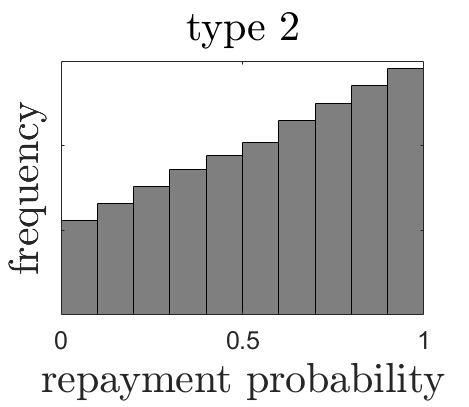}
        \end{minipage} & 
        \begin{minipage}{.19\textwidth}
          \includegraphics[width=\linewidth]{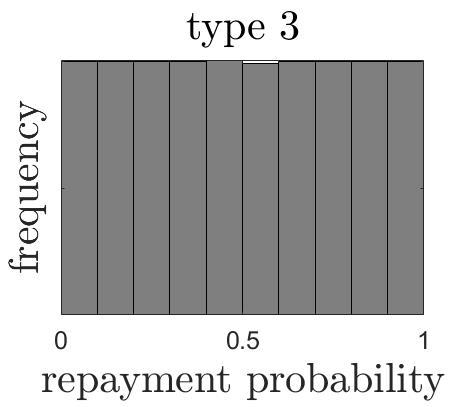}
        \end{minipage} \\\hline
        \begin{tabular}{@{}l@{}} $q_B = \cfrac{1}{n} \sum_{j=1}^n s[j]$ \\ $\pro\left(B = 1\mid \hat{S},M\right) = -\cfrac{1}{16} q_B^2 + \cfrac{1}{2} q_B$ \end{tabular} & 
        \begin{minipage}{.19\textwidth}
          \includegraphics[width=\linewidth]{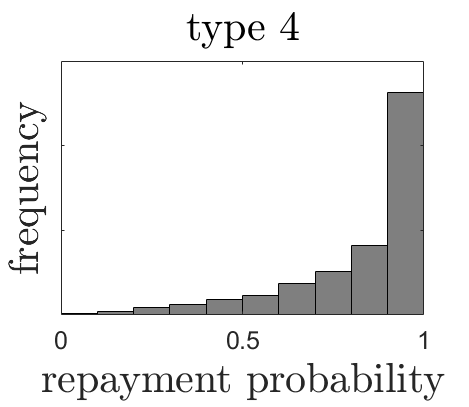}
        \end{minipage} &
        \begin{minipage}{.19\textwidth}
          \includegraphics[width=\linewidth]{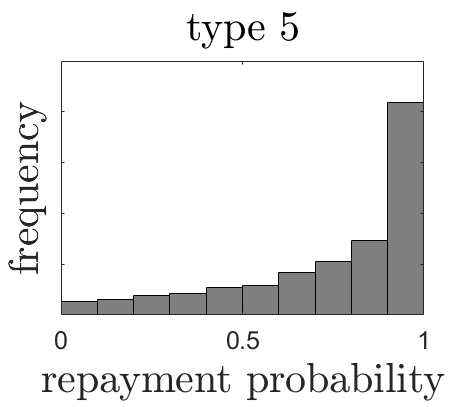}
        \end{minipage} & 
        \begin{minipage}{.19\textwidth}
          \includegraphics[width=\linewidth]{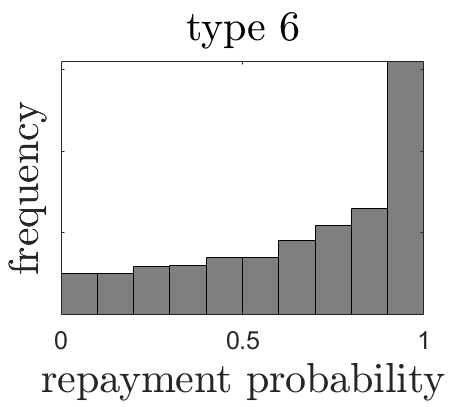}
        \end{minipage} \\\hline
        \begin{tabular}{@{}l@{}} $q_B = \cfrac{1}{n} \sum_{j=1}^n 2 s[j] - 4$ \\ $\pro\left(B = 1\mid \hat{S},M\right) = \cfrac{\exp(q_B)}{1+\exp(q_B)}$ \end{tabular} & 
        \begin{minipage}{.19\textwidth}
          \includegraphics[width=\linewidth]{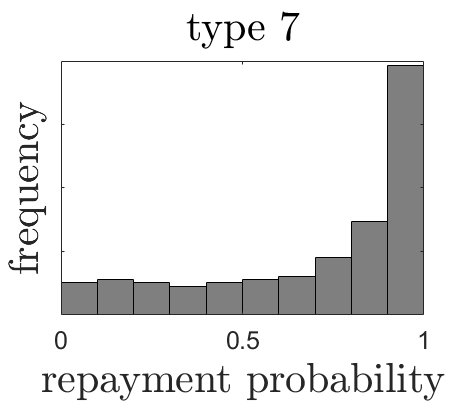}
        \end{minipage} & 
        \begin{minipage}{.19\textwidth}
          \includegraphics[width=\linewidth]{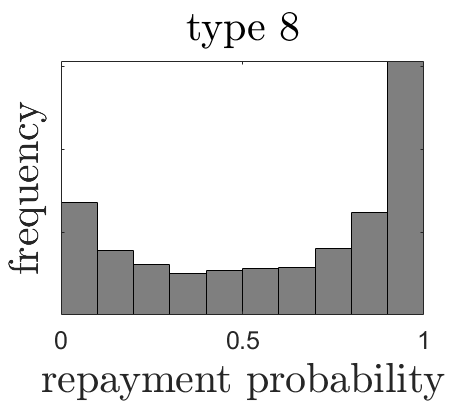}
        \end{minipage} & 
        \begin{minipage}{.19\textwidth}
          \includegraphics[width=\linewidth]{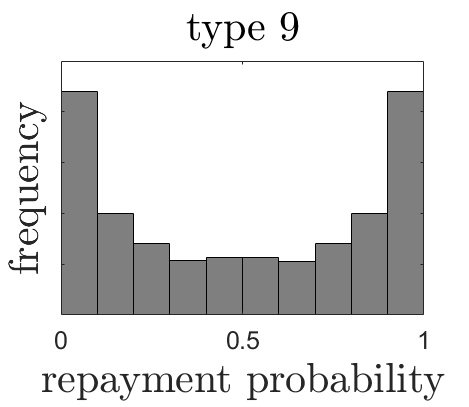}
        \end{minipage} \\\hline
        \begin{tabular}{@{}l@{}} $q_B = \cfrac{1}{n} \sum_{j=1}^n 2.5 s[j] - 4$ \\ $\pro\left(B = 1\mid \hat{S},M\right) = \cfrac{\exp(q_B)}{1+\exp(q_B)}$ \end{tabular} & 
        \begin{minipage}{.19\textwidth}
          \includegraphics[width=\linewidth]{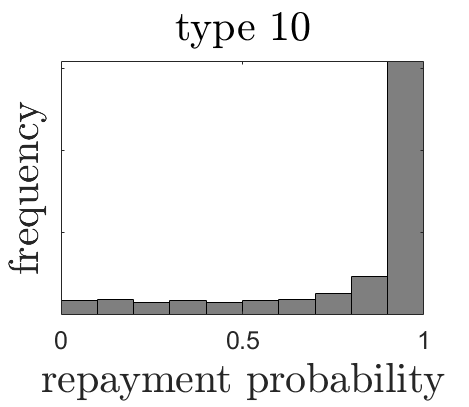}
        \end{minipage} & 
        \begin{minipage}{.19\textwidth}
          \includegraphics[width=\linewidth]{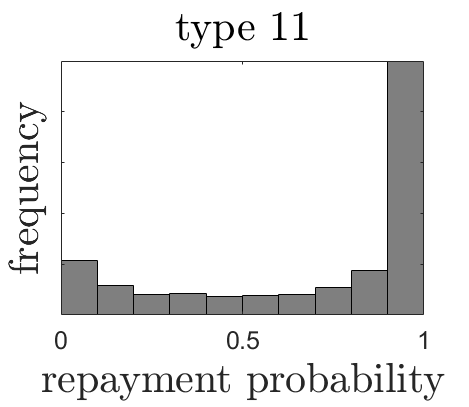}
        \end{minipage} & 
        \begin{minipage}{.19\textwidth}
          \includegraphics[width=\linewidth]{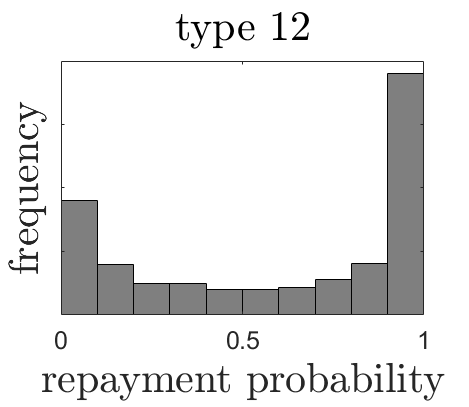}
        \end{minipage} \\\hline
        \begin{tabular}{@{}l@{}} $q_B = \cfrac{1}{n} \sum_{j=1}^n 3 s[j] - 4$ \\ $\pro\left(B = 1\mid \hat{S},M\right) = \cfrac{\exp(q_B)}{1+\exp(q_B)}$ \end{tabular} & 
        \begin{minipage}{.19\textwidth}
          \includegraphics[width=\linewidth]{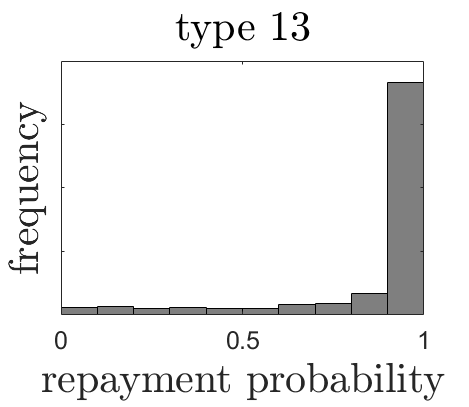}
        \end{minipage} & 
        \begin{minipage}{.19\textwidth}
          \includegraphics[width=\linewidth]{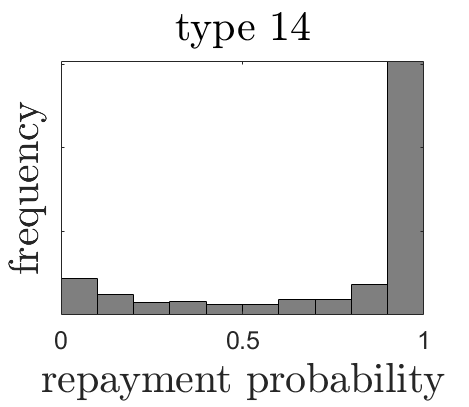}
        \end{minipage} & 
        \begin{minipage}{.19\textwidth}
          \includegraphics[width=\linewidth]{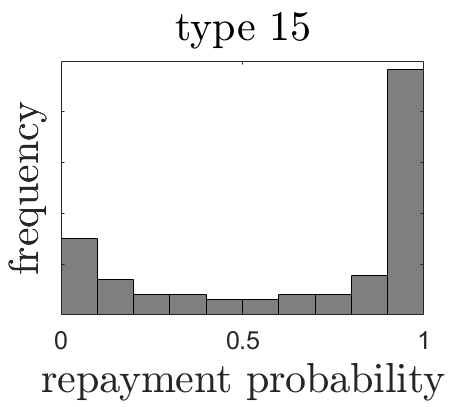}
        \end{minipage} \\\hline
        \begin{tabular}{@{}l@{}} $\{\mathscr{W}[1],...,\mathscr{W}[n]\} \sim \mathcal{N}(2,4^2)$ \\ $q_B = \cfrac{1}{n} \sum_{j=1}^n \mathscr{W}[j] s[j] - 4$ \\ $\pro\left(B = 1\mid \hat{S},M\right) = \cfrac{\exp(q_B)}{1+\exp(q_B)}$ \end{tabular} & 
        \begin{minipage}{.19\textwidth}
          \includegraphics[width=\linewidth]{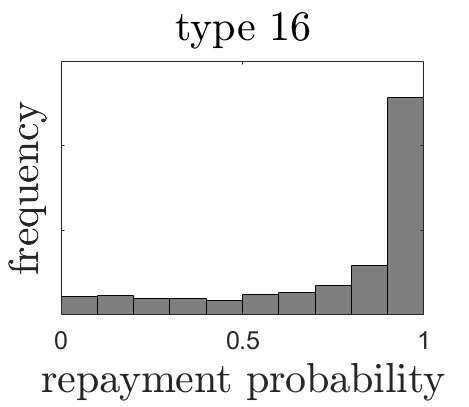}
        \end{minipage} & 
        \begin{minipage}{.19\textwidth}
          \includegraphics[width=\linewidth]{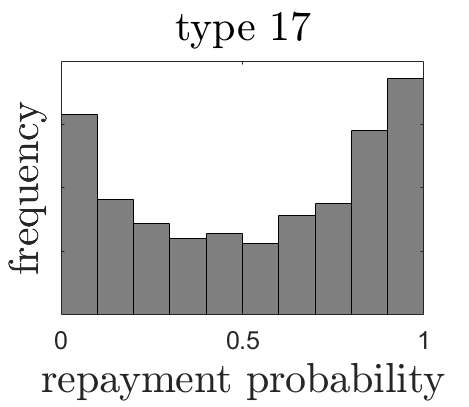}
        \end{minipage} & 
        \begin{minipage}{.19\textwidth}
          \includegraphics[width=\linewidth]{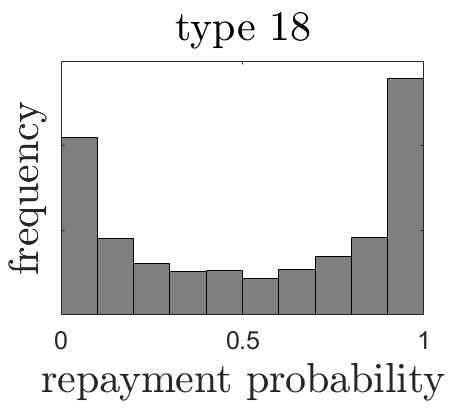}
        \end{minipage} \\\hline
    \end{tabular}
    \caption{List of the repayment probability distributions considered in our study.}
    \label{tab:distribution list}
\end{table*}

\begin{table*}[htbp]
    \centering
    \begin{tabular}{|l||c|c|}\hline
        \multicolumn{1}{|c||}{\backslashbox{repayment \\ probability, \\ $\pro\left(B = 1\mid \hat{S},M\right)$}{features \\ distribution, \\ $\pro \left(S\right)$}} & 
        \begin{tabular}{@{}l@{}}
        \begin{minipage}{.19\textwidth}
          \includegraphics[width=\linewidth]{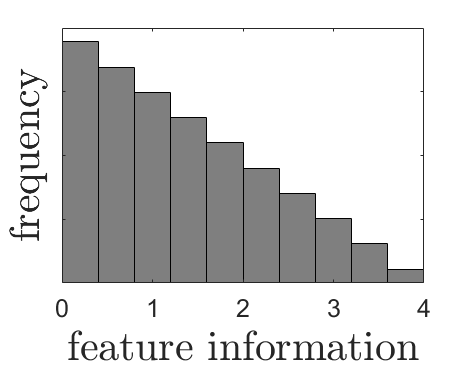}
        \end{minipage} \\
        $\pro\left(S\right)_l = \cfrac{2}{99} - \cfrac{{bin}_l}{4950}$ \\
        \  
        \end{tabular} & 
        \begin{tabular}{@{}l@{}}
        \begin{minipage}{.19\textwidth}
          \includegraphics[width=\linewidth]{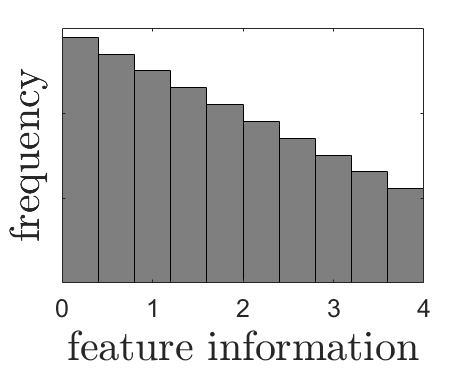}
        \end{minipage} \\
        $\pro\left(S\right)_l = \cfrac{1}{99} - \cfrac{{bin}_l}{9900}$ \\
        \ 
        \end{tabular} \\\hline\hline
        \begin{tabular}{@{}l@{}} $q_B = \cfrac{1}{n} \sum_{j=1}^n - s[j]$ \\ $\pro\left(B = 1\mid \hat{S},M\right) = \cfrac{1}{4}q_B+1$ \end{tabular} & 
        \begin{minipage}{.19\textwidth}
          \includegraphics[width=\linewidth]{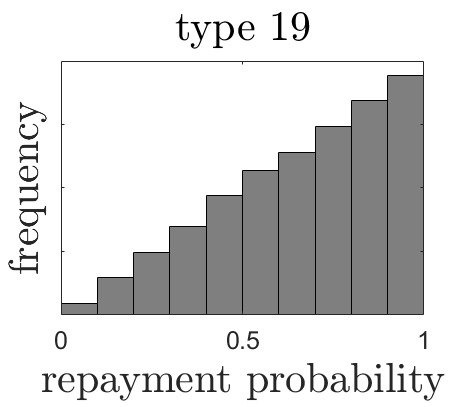}
        \end{minipage} & 
        \begin{minipage}{.19\textwidth}
          \includegraphics[width=\linewidth]{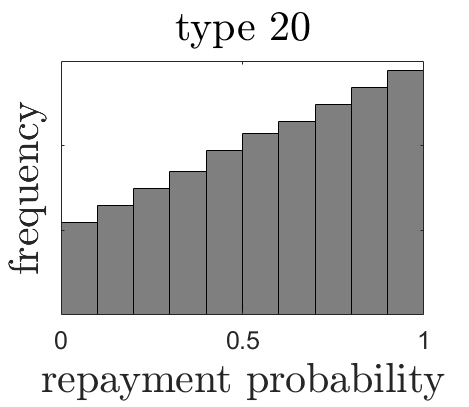}
        \end{minipage} \\\hline
        \begin{tabular}{@{}l@{}} $q_B = \cfrac{1}{n} \sum_{j=1}^n -2 s[j]$ \\ $\pro\left(B = 1\mid \hat{S},M\right)$\\
        $ = -\cfrac{1}{128} q_B^2 + \cfrac{1}{16} q_B + 1$ \end{tabular} &
        \begin{minipage}{.19\textwidth}
          \includegraphics[width=\linewidth]{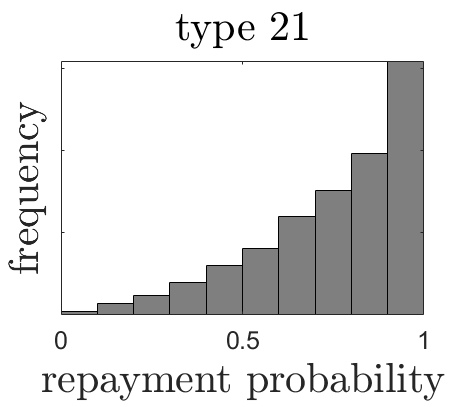}
        \end{minipage} &
        \begin{minipage}{.19\textwidth}
          \includegraphics[width=\linewidth]{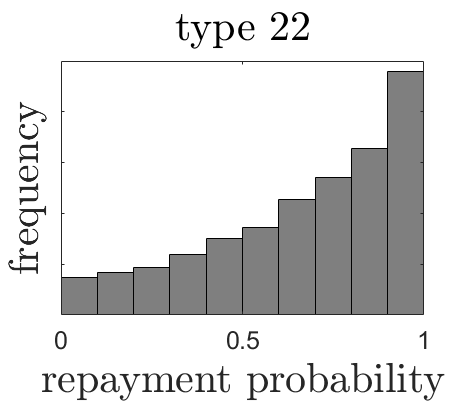}
        \end{minipage} \\\hline
        \begin{tabular}{@{}c@{}} $q_B = \cfrac{1}{n} \sum_{j=1}^n -3 s[j] + 7$ \\ $\pro\left(B = 1\mid \hat{S},M\right) = \cfrac{\exp(q_B)}{1+\exp(q_B)}$ \end{tabular} & 
        \begin{minipage}{.19\textwidth}
          \includegraphics[width=\linewidth]{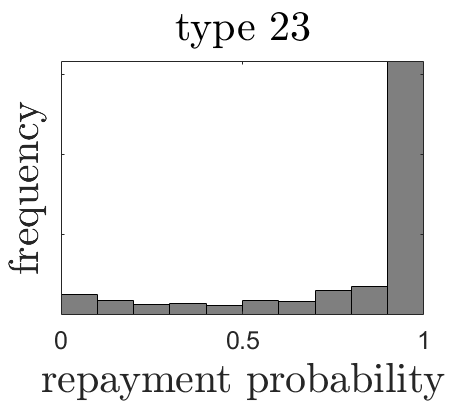}
        \end{minipage} & 
        \begin{minipage}{.19\textwidth}
          \includegraphics[width=\linewidth]{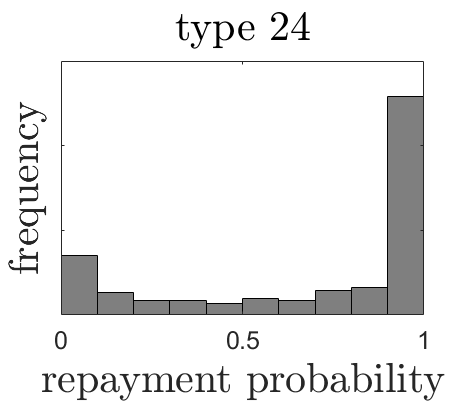}
        \end{minipage} \\\hline
        \begin{tabular}{@{}c@{}} $\{\mathscr{W}[1],...,\mathscr{W}[n]\} \sim \mathcal{N}(-3,4^2)$ \\ $q_B = \cfrac{1}{n} \sum_{j=1}^n \mathscr{W}[j] s[j] + 7$ \\ $\pro\left(B = 1\mid \hat{S},M\right) = \cfrac{\exp(q_B)}{1+\exp(q_B)}$ \end{tabular} & 
        \begin{minipage}{.19\textwidth}
          \includegraphics[width=\linewidth]{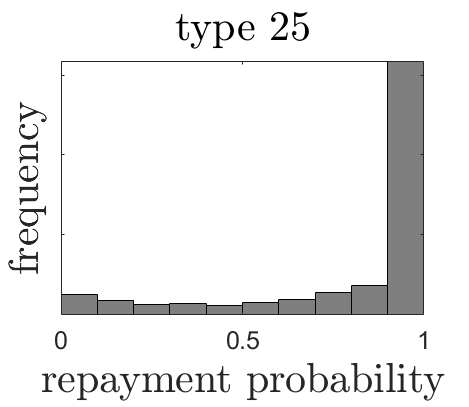}
        \end{minipage} & 
        \begin{minipage}{.19\textwidth}
          \includegraphics[width=\linewidth]{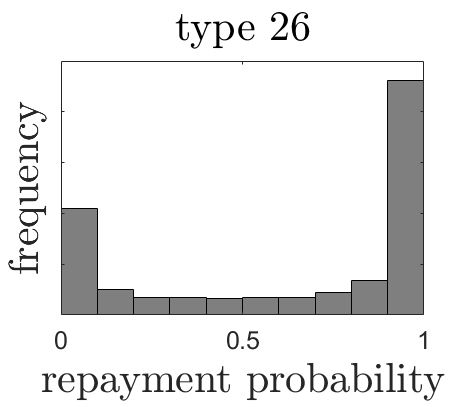}
        \end{minipage} \\\hline
    \end{tabular}
    \caption{List of the repayment probability distributions considered in our study with more negative weights for the features.}
    \label{tab:distribution list flip}
\end{table*}

\begin{table*}[htbp]
    \centering
    \begin{tabular}{|c||c|l|}\hline
        \begin{tabular}{@{}c@{}} features distribution, \\ $\pro (S)$ \end{tabular} & 
        \multicolumn{2}{c|}{repayment probability, $\pro\left(B = 1\mid \hat{S},M\right)$} \\\hline\hline
        \begin{tabular}{@{}c@{}}
        \begin{minipage}{.19\textwidth}
          \includegraphics[width=\linewidth]{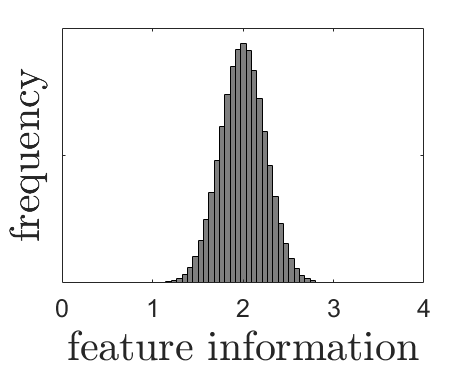}
        \end{minipage} \\
        $S \sim \mathcal{N}\left(2,0.25^2\right)$
        \end{tabular} & 
        \begin{minipage}{.19\textwidth}
          \includegraphics[width=\linewidth]{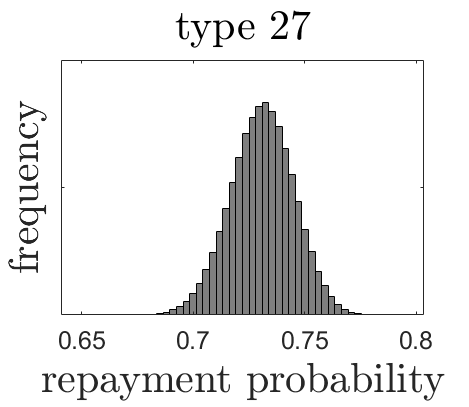}
        \end{minipage} &
        \begin{tabular}{@{}l@{}} $\mathscr{W}[j] = \cfrac{5}{99}(j-1)$ \\ $q_B = \cfrac{1}{n} \sum_{j=1}^n \mathscr{W}[j] s[j] - 4$ \\ $\pro\left(B = 1\mid \hat{S},M\right) = \cfrac{\exp(q_B)}{1+\exp(q_B)}$ \end{tabular} \\\hline
        \begin{tabular}{@{}c@{}}
        \begin{minipage}{.19\textwidth}
          \includegraphics[width=\linewidth]{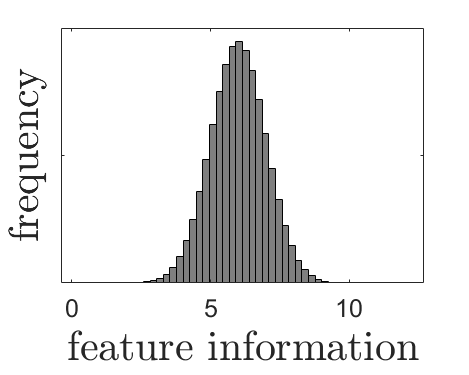}
        \end{minipage} \\
        $S \sim \mathcal{N}\left(6,1^2\right)$
        \end{tabular} & 
        \begin{minipage}{.19\textwidth}
          \includegraphics[width=\linewidth]{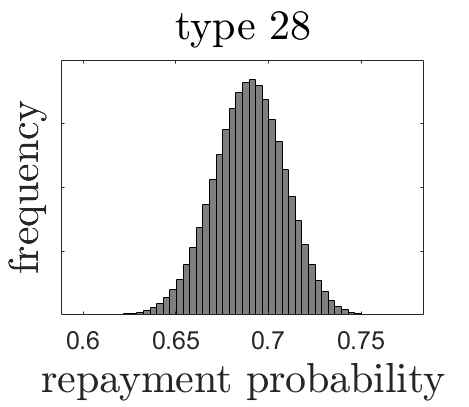}
        \end{minipage} &
        \begin{tabular}{@{}l@{}} $\mathscr{W}[j] = \begin{cases} 1.5 - \cfrac{7}{495}\left(j-1\right), & j \leq \cfrac{n}{2} \\ 0.1 + \cfrac{7}{495}\left(j-\cfrac{n}{2}-1\right), & j > \cfrac{n}{2} \end{cases}$ \\ $q_B = \cfrac{1}{n} \sum_{j=1}^n \mathscr{W}[j] s[j] - 4$ \\ $\pro\left(B = 1\mid \hat{S},M\right) = \cfrac{\exp(q_B)}{1+\exp(q_B)}$ \end{tabular} \\\hline
        \begin{tabular}{@{}c@{}}
        \begin{minipage}{.19\textwidth}
          \includegraphics[width=\linewidth]{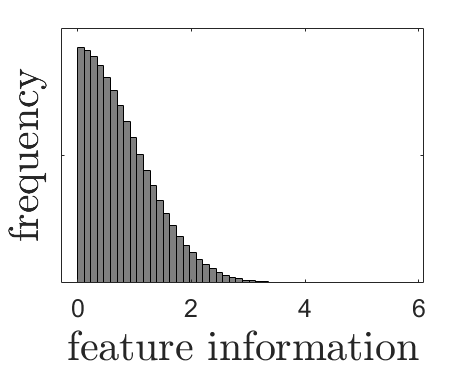}
        \end{minipage} \\
        $S \sim \left|\mathcal{N}\left(0,1^2\right)\right|$
        \end{tabular} & 
        \begin{minipage}{.19\textwidth}
          \includegraphics[width=\linewidth]{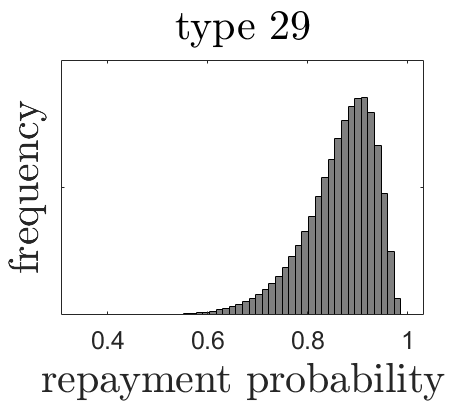}
        \end{minipage} &
        \begin{tabular}{@{}l@{}} $\mathscr{W}[j] = 20 - \cfrac{25}{99}(j-1)$ \\ $q_B = \cfrac{1}{n} \sum_{j=1}^n \mathscr{W}[j] s[j] - 4$ \\ $\pro\left(B = 1\mid \hat{S},M\right) = \cfrac{\exp(q_B)}{1+\exp(q_B)}$ \end{tabular} \\\hline
        \begin{tabular}{@{}c@{}}
        \begin{minipage}{.19\textwidth}
          \includegraphics[width=\linewidth]{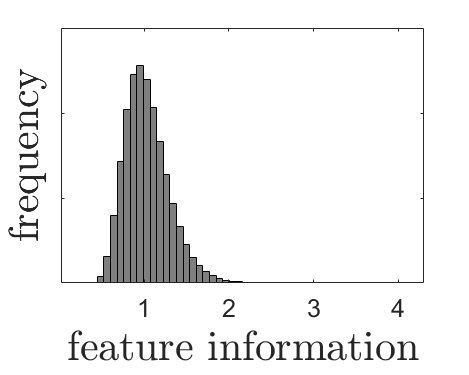}
        \end{minipage} \\
        $S \sim \exp\left(\mathcal{N}\left(0,0.25^2\right)\right)$
        \end{tabular} & 
        \begin{minipage}{.19\textwidth}
          \includegraphics[width=\linewidth]{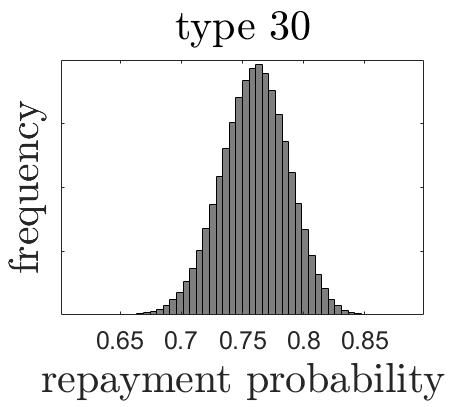}
        \end{minipage} &
        \begin{tabular}{@{}l@{}} $\mathscr{W}[j] = \cfrac{10}{99}(j-1)$ \\ $q_B = \cfrac{1}{n} \sum_{j=1}^n \mathscr{W}[j] s[j] - 4$ \\ $\pro\left(B = 1\mid \hat{S},M\right) = \cfrac{\exp(q_B)}{1+\exp(q_B)}$ \end{tabular} \\\hline
    \end{tabular}
    \caption{List of the repayment probability distributions considered in our study with unbounded features distributions.}
    \label{tab:distribution list unbounded}
\end{table*}

\end{document}